\documentclass{article}

\usepackage{geometry}
\usepackage{graphicx}
\usepackage{subfigure}
\usepackage{amssymb,amsmath,amsthm}

\newcommand{\R}{\mathbb{R}}
\newcommand{\Q}{\mathbb{Q}}
\newcommand{\N}{\mathbb{N}}

\newtheorem{theorem}{Theorem}
\newtheorem{lemma}{Lemma}
\newtheorem{claim}{Claim}
\newtheorem{conjecture}{Conjecture}

\begin{document}

\title{The Complexity of Counting Eulerian Tours in 4-Regular Graphs\thanks{Research supported, in part, by NSF grant CCF-0910584.
This paper is an extension of the previous work presented at the $9$th Latin American Theoretical Informatics Symposium (LATIN 2010)~\cite{DBLP:conf/latin/GeS10}.}
}

\author{Qi Ge\thanks{ Department of Computer Science, University
of Rochester, Rochester, NY 14627.  Email: \{qge,stefanko\}@cs.rochester.edu.} \and Daniel
\v{S}tefankovi\v{c}$^{\dag}$}

\maketitle

\begin{abstract}
We investigate the complexity of counting Eulerian tours ({\sc
\#ET}) and its variations from two perspectives---the complexity
of exact counting and the complexity w.r.t.
approximation-preserving reductions
(AP-reductions~\cite{MR2044886}). We prove that {\sc \#ET} is
\#P-complete even for planar $4$-regular graphs.

A closely related
problem is that of counting A-trails ({\sc \#A-trails}) in graphs
with rotational embedding schemes (so called maps).
Kotzig~\cite{MR0248043} showed that {\sc \#A-trails} can be
computed in polynomial time for $4$-regular plane graphs
(embedding in the plane is equivalent to giving a rotational
embedding scheme). We show that for $4$-regular maps the problem
is \#P-hard. Moreover, we show that from the approximation
viewpoint {\sc \#A-trails} in $4$-regular maps captures the
essence of {\sc \#ET}, that is, we give an AP-reduction from {\sc
\#ET} in general graphs to {\sc \#A-trails} in $4$-regular maps.
The reduction uses a fast mixing result for a card shuffling
problem~\cite{MR2023023}.

In order to understand whether \#{\sc A-trails} in $4$-regular
maps can AP-reduce to \#{\sc ET} in $4$-regular graphs, we
investigate a problem in which transitions in vertices
are weighted (this generalizes both \#{\sc A-trails} and
\#{\sc ET}). In the $4$-regular case we show that {\sc A-trails}
can be used to simulate any vertex weights and provide evidence
that {\sc ET} can simulate only a limited set of vertex weights.
\end{abstract}

\section{Introduction}

An Eulerian tour in a graph is a tour which travels each edge
exactly once. The problem of counting Eulerian tours ({\sc \#ET}) of
a graph is one of a few recognized counting problems (see, e.\,g.,
\cite{MR1851303}, p. 339). The exact counting is
\#P-complete in general graphs~\cite{DBLP:conf/alenex/BrightwellW05} and in planar graphs~\cite{creedthesis}, and thus there
is no polynomial-time algorithm for it unless P$=$NP. For the
approximate counting one wants to have a fully polynomial randomized
approximation scheme (FPRAS), that is, an algorithm which on every
instance $x$ of the problem and error parameter $\varepsilon
> 0$, will output a value
within a factor $\exp(\pm\varepsilon)$ of $f(x)$ with probability at
least $2/3$ and in time polynomial in the length of the encoding of
$x$ and $1/\varepsilon$, where $f(x)$ is the value we want to compute.
The existence of an FPRAS for {\sc \#ET} is an open
problem~\cite{MR1822924,Jerrum,MR1851303}.

A closely related problem to {\sc \#ET} is the problem of counting
A-trails ({\sc \#A-trails}) in graphs with rotational embedding
schemes (called maps, see Section~\ref{sec:def} for a definition).
A-trails were studied in the context of decision problems (for
example, it is NP-complete to decide whether a given plane graph
has an A-trail~\cite{MR905181,MR1331474}; on the other hand for
4-regular maps the problem is in P~\cite{dvorak-eulerian}), as
well as counting problems (for example, Kotzig~\cite{MR0248043}
showed that {\sc \#A-trails} can be computed in polynomial time
for $4$-regular plane graphs, reducing the problem to
counting of spanning trees).

In this paper, we investigate the complexity of {\sc \#ET} in 4-regular graphs
and its variations from two perspectives. First, the complexity of exact counting is considered.
We prove that {\sc \#ET} in 4-regular graphs (even in 4-regular planar
graphs) is \#P-complete. We also prove that {\sc \#A-trails} in 4-regular
maps is \#P-complete (recall that
the problem can be solved in polynomial time for 4-regular {\em plane} graphs).

The second perspective is the complexity w.r.t. the AP-reductions
proposed by Dyer, Goldberg, Greenhill and Jerrum~\cite{MR2044886}.
We give an AP-reduction from {\sc \#ET} in general graphs
to {\sc \#A-trails} in 4-regular maps. Thus we show that if there is an FPRAS for {\sc \#A-trails}
in 4-regular maps, then there is also an FPRAS for {\sc \#ET} in
general graphs. The existence of AP-reduction from {\sc \#ET} in general graphs to {\sc
\#ET} in 4-regular graphs is left open.

In order to understand whether \#{\sc A-trails} in $4$-regular maps can AP-reduce to
\#{\sc ET} in $4$-regular graphs, we investigate the so-called signatures (these count
connection patterns of trails in graphs with half-edges, see Section~\ref{sec:signature}
for the formal definition) of $4$-regular map gadgets
and $4$-regular graph gadgets. It seems that the signatures represented by $4$-regular map gadgets form a
proper superset of the set of signatures represented by $4$-regular graph gadgets. Moreover, it seems that
the signature of a single vertex in $4$-regular maps cannot be simulated approximately by $4$-regular graph gadgets.

\section{Definitions and Terminology}\label{sec:def}

For the definitions of cyclic orderings, A-trails, and mixed
graphs, we follow~\cite{MR1055084}. Let $G=(V,E)$ be a graph. For
a vertex $v \in V$ of degree $d > 0$, let $K(v)=\{e_1,\ldots,e_d\}$
be the set of edges adjacent to $v$ in $G$. The \emph{cyclic ordering}
$O^+(v)$ of the edges adjacent to $v$ is a $d$-tuple
$(e_{\sigma(1)},\ldots,e_{\sigma(d)})$, where $\sigma$ is a
permutation in $S_d$. We say $e_{\sigma(i)}$ and $e_{\sigma(i+1)}$
are cyclicly-adjacent in $O^+(v)$, for $1 \leq i \leq d$, where we
set $\sigma(d+1):=\sigma(1)$. The set $O^+(G)=\{O^+(v)|v \in
V\}$ is called a \emph{rotational embedding scheme} of~$G$. For a
plane graph $G=(V,E)$, if $O^+(v)$ is not specified, we usually
set $O^+(v)$ to be the clockwise order of the half-edges adjacent
to $v$ for each $v \in V$.

Let $G=(V,E)$ be a graph with a rotational embedding scheme
$O^+(G)$. An Eulerian tour $v_0,e_1,v_1,e_2,\ldots,e_\ell,v_\ell=v_0$
is called an \emph{A-trail} if $e_i$ and $e_{i+1}$ are cyclicly-adjacent
in $O^+(v_i)$, for each $1 \leq i \leq \ell$, where we set
$e_{\ell+1}:=e_1$.

Let $G=(V,E,E')$ be a mixed graph, that is, $E$ is the set of edges
and $E'$ is the set of half-edges (which are incident with only one
vertex in $V$). Let $|E'|=2d$ where $d$ is a positive integer and
assume that the half-edges in $E'$ are labelled by numbers from $1$
to $2d$. A \emph{route} $r(a,b)$ is a trail (no repeated edges,
repeated vertices allowed) in $G$ that starts with half-edge $a$ and
ends with half-edge $b$. A collection of $d$ routes is called {\em
valid} if every edge and every half-edge is travelled exactly once.

We say that a valid set of routes
is of the {\em type} $\{\{a_1,b_1\},\ldots,\{a_{d},b_{d}\}\}$
if it contains routes connecting $a_i$ to $b_i$ for $i\in [d]$.
We use ${\rm VR}(\{a_1,b_1\},\ldots,\{a_{d},b_{d}\})$ to denote the set of valid
sets of routes of type $\{\{a_1,b_1\},\ldots,\{a_{d},b_{d}\}\}$ in
$G$.

We will use the following concepts from Markov chains to construct
the gadget in Section~\ref{sec:app} (see, e.\,g., ~\cite{MR1960003}
for more detail).
Given two probability distributions $\pi$ and $\pi'$ on finite
set $\Omega$, the \emph{total variation distance} between $\pi$ and
$\pi'$ is defined as
$$\|\pi-\pi'\|_{TV}=\frac{1}{2}\sum_{\omega \in \Omega}|\pi(\omega)-\pi'(\omega)|=\max_{A \subseteq \Omega}|\pi(A)-\pi'(A)|.$$

Given a finite ergodic Markov chain with transition matrix $P$ and
stationary distribution $\pi$, the \emph{mixing time} from initial
state $x$, denoted as $\tau_x(\varepsilon)$, is defined as
$$\tau_x(\varepsilon)=\min\{t:\|P^t(x,\cdot)-\pi\|_{TV} \leq \varepsilon\},$$
and the mixing time of the chain $\tau(\varepsilon)$ is defined as
$$\tau(\varepsilon)=\max_{x \in \Omega}\{\tau_x(\varepsilon)\}.$$

\section{The complexity of exact counting}\label{sec:exact}

\subsection{Basic gadgets}\label{sec:gadget}

We describe two basic gadgets and their properties which will be
used as a basis for larger gadgets in the subsequent sections.

The first gadget, which is called the $(X,Y,Y)$ node, is shown in
Figure~\ref{fig:node_xyy}, and it is represented by the symbol shown
in Figure~\ref{fig:represent_xyy}. There are $k$ internal vertices
in the gadget, and the labels 0, 1, 2 and 3 are four half-edges of
the $(X,Y,Y)$ node which are the only connections from the outside.

\begin{figure}[htb]
\begin{center}
 \subfigure[]{\label{fig:node_xyy}\includegraphics[scale=0.7]{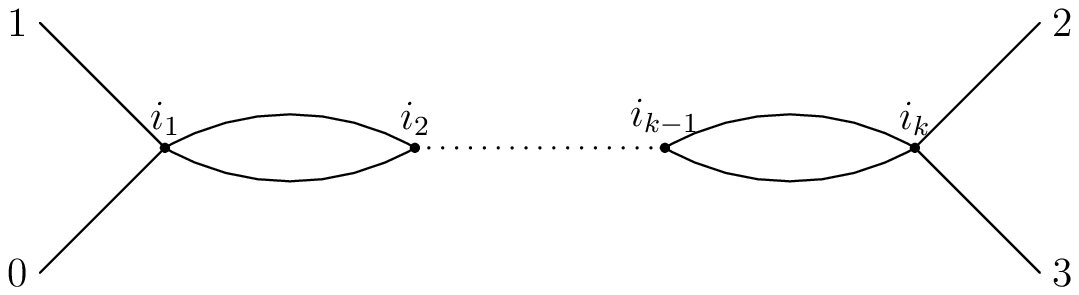}}\hspace{1cm}
 \subfigure[]{\label{fig:represent_xyy}\includegraphics[scale=0.4]{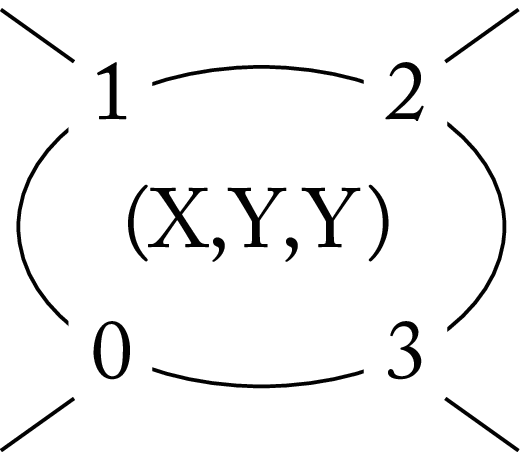}}
\end{center}
\caption{An $(X,Y,Y)$ node and its symbol.
\subref{fig:node_xyy}: an $(X,Y,Y)$ node consisting of $k$ internal
vertices; \subref{fig:represent_xyy}: symbol representing the
$(X,Y,Y)$ node}
\end{figure}

By elementary counting we obtain the following fact.
\begin{lemma}\label{prop:node_xyy}
The $(X,Y,Y)$ node with parameter $k$ has three different types of valid sets of routes
and these satisfy
\begin{eqnarray*}
|{\rm VR}(\{0,1\},\{2,3\})|&=&k 2^{k-1},\\
|{\rm VR}(\{0,2\},\{1,3\})|=|{\rm VR}(\{0,3\},\{1,2\})|&=& 2^{k-1}.
\end{eqnarray*}
The gadget has $k$ vertices.
\end{lemma}

The second gadget, which is called the $(0,X,Y)$ node, is shown in
Figure~\ref{fig:node_0xy}, and it is represented by the symbol
shown in Figure~\ref{fig:represent_0xy}. Let $p$ be any odd prime.
In the construction of the $(0,X,Y)$ node we use $p$ copies of
$(X,Y,Y)$ nodes as basic components, and each $(X,Y,Y)$ node has
the same parameter $k$. As illustrated, half-edges are connected
between two consecutive $(X,Y,Y)$ nodes. The four labels 0, 1, 2
and 3 at four corners in Figure~\ref{fig:node_0xy} are the four
half-edges of the $(0,X,Y)$ node, and they are the only
connections from the outside.

\begin{figure}[htb]
\begin{center}
 \subfigure[]{\label{fig:node_0xy}\includegraphics[scale=0.33]{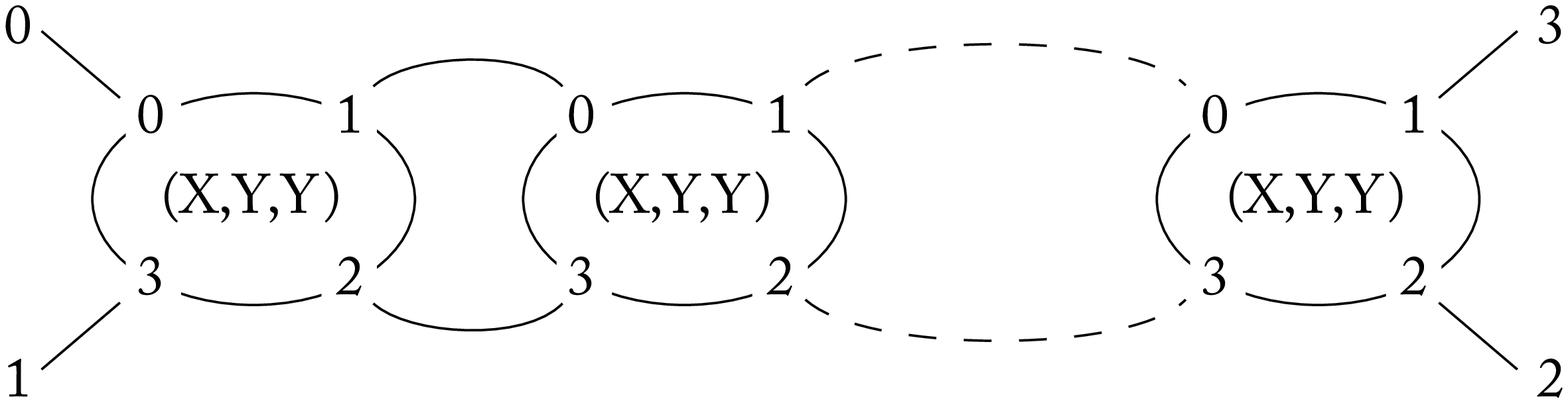}}\hspace{1cm}
 \subfigure[]{\label{fig:represent_0xy}\includegraphics[scale=0.4]{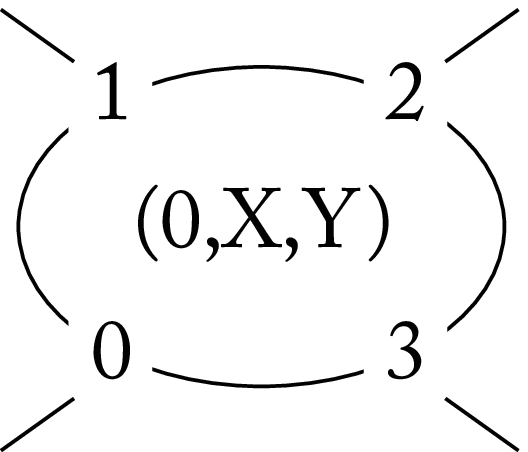}}
\end{center}
\caption{A $(0,X,Y)$ node and its symbol.
\subref{fig:node_0xy}: a $(0,X,Y)$ node consisting of $p$ copies of
$(X,Y,Y)$ nodes; \subref{fig:represent_0xy}: symbol representing the
$(0,X,Y)$ node}
\end{figure}

By elementary counting, binomial expansion, Fermat's little theorem, and the fact
that $2$ has a multiplicative inverse mod $p$, we obtain the following:
\begin{lemma}\label{prop:node_0xy}
Let $p$ be an odd prime and let $k$ be an integer. The $(0,X,Y)$ node with
parameters $p$ and $k$ has three different types of valid sets of routes
and these satisfy
\begin{eqnarray}
|{\rm VR}(\{0,1\},\{2,3\})|&=& pA(A+B)^{p-1} \equiv 0 \mod p , \label{ee1}\\
|{\rm VR}(\{0,2\},\{1,3\})|&=& \frac{(A+B)^p-(B-A)^p}{2} \equiv A \mod p,\label{ee2}\\
|{\rm VR}(\{0,3\},\{1,2\})|&=& \frac{(A+B)^p+(B-A)^p}{2} \equiv B \mod p\label{ee3},
\end{eqnarray}
where $A=2^{k-1}$ and $B=k 2^{k-1}$. The gadget has $kp$ vertices.
\end{lemma}

\subsection{{\sc \#ET} in 4-regular graphs is \#P-complete}

Next, we will give a reduction from {\sc \#ET} in general Eulerian graphs to {\sc \#ET} in 4-regular graphs.
\begin{theorem}\label{ttt}
{\sc \#ET} in general Eulerian graphs is polynomial time Turing reducible
to {\sc \#ET} in 4-regular graphs.
\end{theorem}

The proof of Theorem~\ref{ttt} is postponed to the end of this section.

We use the gadget, which we will call $Q$, illustrated in
Figure~\ref{fig:gadget_gto4_p} to prove the Theorem. The gadget is
constructed in a recursive way. The $d$ labels $1,\ldots,d$ on the
left are called input half-edges of the gadget, and the $d$ labels
on the right are called output half-edges. Given a prime $p$ and a
positive integer $d$, the gadget consists of $d-1$ copies of
$(0,X,Y)$ nodes with different parameters and one recursive part
represented by a rectangle with $d-1$ input half-edges and $d-1$
output ones. For $1 \leq i \leq d-1$, the $i$-th $(0,X,Y)$ node
from left has parameters $p$ and $i$. Half-edge 0 of the $i$-th
$(0,X,Y)$ node is connected to half-edge 3 of the $(i-1)$-st
$(0,X,Y)$ node except that for the 1st $(0,X,Y)$ node half-edge 0
is the $d$-th input half-edge of the gadget. Half-edge 1 of the
$i$-th $(0,X,Y)$ node is the $(d-i)$-th input half-edge of the
gadget. Half-edge 2 of the $i$-th $(0,X,Y)$ node is connected to
the $(d-i)$-th input half-edge of the rectangle. Half-edge 3 of
the $(d-1)$-st $(0,X,Y)$ node is the $d$-th output half-edge of
the gadget. For $1 \leq j \leq d-1$, the $j$-th output half-edge
of the rectangle is the $j$-th output half-edge of the gadget.
From the constructions of $(X,Y,Y)$ nodes and $(0,X,Y)$ nodes, the
total size of the $d-1$ copies of $(0,X,Y)$ nodes is $O(pd^2)$.
Thus, the size of the gadget is $O(pd^3)$.

\begin{figure}[htb]
\centering
\includegraphics[scale=0.33]{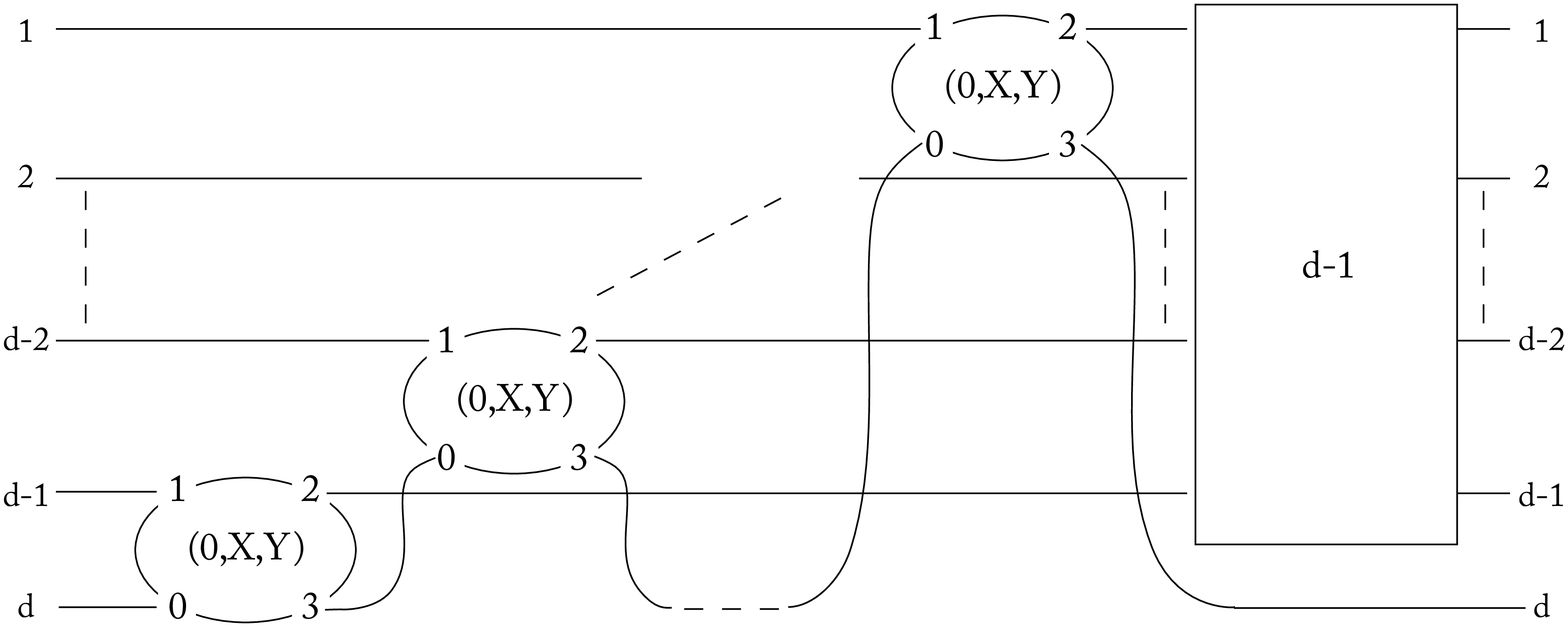}
\caption{Gadget $Q$ with $d$ input half-edges and $d$ output
half-edges}\label{fig:gadget_gto4_p}
\end{figure}

\begin{lemma}\label{le3}
Consider the gadget $Q$ with parameters $d$ and $p$. Let $\sigma$ be a permutation
in $S_d$. Then
\begin{equation}\label{aaa}
|{\rm VR}(\sigma)|:=|{\rm VR}(\{IN_1,OUT_{\sigma(1)}\},\ldots,\{IN_d,OUT_{\sigma(d)}\})|\equiv R_d \mod p,\
\end{equation}
where $R_d \equiv \prod_{i=1}^{d-1}(2^{i(i-1)/2}i!)$.

Moreover, any type $\tau$ which connects two IN (or two OUT) half-edges satisfies
\begin{equation}\label{bbb}
|{\rm VR}(\tau)|\equiv 0\mod p.
\end{equation}
\end{lemma}

\begin{proof}
The proof is by induction on $d$, the base case $d=1$ is trivial.
Suppose the statement is true for gadget $Q$ with $(d-1)$ input
half-edges, that is, $|{\rm VR}(\varrho)| \equiv R_{d-1} \mod p$ for
every $\varrho \in S_{d-1}$.

Now, consider gadget $Q$ with $d$ input half-edges. For $1 \leq j
\leq d-1$, we cut the gadget by a vertical line just after the $j$-th
$(0,X,Y)$ node and only consider the part of the gadget to the left
of the line, we will call this partial gadget $Q_j$.

\begin{claim}
Let $A_s$ be the set of permutations in $S_d$ which map
$s$ to $d$. In the partial gadget $Q_j$ we have that for
$s\in\{d-j,\dots,d\}$ have
$$
\sum_{\sigma\in A_s} |{\rm VR}_{Q_j}(\sigma)|\equiv j!2^{j(j-1)/2} \mod p,
$$
where the subscript $Q_j$ is used to indicate that we count routes in gadget $Q_j$.
\end{claim}

\begin{proof}[Proof of Claim]
We prove the claim by induction on $j$, the base case $j=1$ is trivial.

Now assume that the claim is true for $j-1$, that is, for all $s\in\{d-j+1,\dots,d\}$
in gadget $Q_{j-1}$ we have
$$
\sum_{\sigma\in A_s} |{\rm VR}_{Q_{j-1}}(\sigma)|\equiv (j-1)!2^{(j-1)(j-2)/2} \mod p.
$$
The $j$-th $(0,X,Y)$ node takes $(d-j)$-th input
half-edge of the gadget and the half-edge 3 of the $(j-1)$-st $(0,X,Y)$
node, and has parameters $p$ and $j$.

The type of the $j$-th $(0,X,Y)$ node is $\{\{0,2\},\{1,3\}\}$ if and
only if the resulting permutation in $Q_j$ is in $A_{d-j}$.  Thus we have
$$
\sum_{\sigma\in A_{d-j}} |{\rm VR}_{Q_j}(\sigma)| \equiv 2^{j-1} \prod_{k=1}^{j-1} (2^{k-1}(k+1)) \equiv j! 2^{j(j-1)/2} \mod p,
$$
where the first term is the number of choices (modulo $p$) in the $j$-th $(0,X,Y)$ node to make it
$\{\{0,2\},\{1,3\}\}$ and the $k$-th term in the product is the number of choices (modulo $p$)
in the $k$-th  $(0,X,Y)$ node to make it either $\{\{0,2\},\{1,3\}\}$ or
$\{\{0,3\},\{1,2\}\}$.

If the type inside the $j$-th $(0,X,Y)$ node is $\{\{0,3\},\{1,2\}\}$ then the resulting permutation
is in $A_s$ for $s\in\{d-j+1,\dots,d\}$. Thus
\begin{align*}
\sum_{\sigma\in A_{s}} |{\rm VR}_{Q_j}(\sigma)| \equiv j 2^{j-1} \sum_{\sigma\in A_{s}} |{\rm VR}_{Q_{j-1}}(\sigma)|
&\equiv j 2^{j-1} (j-1)!2^{(j-1)(j-2)/2}\\
&\equiv j! 2^{j(j-1)/2} \mod p,
\end{align*}
where $j2^{j-1}$ is the number of choices (modulo $p$) in the $j$-th $(0,X,Y)$ node to make
it $\{\{0,3\},\{1,2\}\}$.
\end{proof}

Now we continue with the proof of the Lemma~\ref{le3}.

Let $\sigma$ be a permutation in $S_d$. Let $l=\sigma^{-1}(d)$. In order for $\sigma$
to be realized by gadget $Q$ we have to have $l$ mapped to $d$ by $Q_{d-1}$ and
the permutation realized by the recursive gadget of size $d-1$ must ``cancel'' the
permutation of $Q_{d-1}$. By the claim there are $(d-1)! 2^{(d-1)(d-2)/2}$ (modulo $p$)
choices in $Q_{d-1}$ which map $l$ to $d$ and by the inductive hypothesis there
are $R_{d-1}$ (modulo $p$) choices in the recursive gadget of size $d-1$ that give the
unique permutation that ``cancels'' the
permutation of $Q_{d-1}$. Thus
$$|{\rm VR}(\sigma)|\equiv R_d \equiv (d-1)! 2^{(d-1)(d-2)/2} R_{d-1}\mod{p},$$
finishing the proof of~\eqref{aaa}.

To see~\eqref{bbb} note that the number of valid sets of routes which
contain route starting and ending at both input half-edges or both
output half-edges is 0 modulo $p$. This is because the number of
valid set of routes of type $\{\{0,1\},\{2,3\}\}$ inside the $(0,X,Y)$ node
is 0 modulo $p$. 
\end{proof}

\begin{proof}[Proof of Theorem~\ref{ttt}]
The reduction is now a standard application of the Chinese remainder
theorem. Given an Eulerian graph $G=(V,E)$, we can, w.l.o.g., assume
that the degree of vertices of $G$ is at least $4$ (vertices of
degree $2$ can be removed by contracting edges). The number of
Eulerian tours of a graph on $n$ vertices is bounded by $n^{n^2}$
(the number of pairings in a vertex of degree $d$ is
$d!/(2^{d/2}(d/2)!)\leq n^n$).

We choose $n^2$ primes $p_1,\ldots,p_{n^2}>n$ such that
$\prod_{i=1}^{n^2} p_i > n^{n^2}$ and each $p_i$ is bounded by
$O(n^3)$ (see, e.\,g., ~\cite{MR1406794}, p.296). For each $p_i$, we
construct graph $G_i$ by replacing each vertex $v$ of degree $d > 4$
with $Q$ gadget with $d$ input and $d$ output half-edges where the
$(2j-1)$-st and $2j$-th output half-edge are connected (for
$j=1,\dots,d/2$), and the input half-edges are used to replace
half-edges emanating from $v$ (that is, they are connected to the
input half-edges of other gadgets according to the edge incidence at
$v$). Note that $G_i$ is a 4-regular graph. Since $p_i=O(n^3)$, the
construction of $G_i$ can be done in time polynomial in $n$. Having
$G_i$, we make a query to the oracle and obtain the number $T_i$ of
Eulerian tours in $G_i$. Let $T$ be the number of Eulerian tours in
$G$. Then
\begin{equation}\label{ccc}
T_i \equiv T\prod_{d =
6}^{n}\left(\left(\frac{d}{2}\right)!2^{d/2}R_d
\right)^{n_d} \mod p_i,
\end{equation}
where $n_d$ is the number of vertices of degree $d$ in $G$.

Since $T_i$ is of length polynomial in $n$, we can compute $T_i \mod
p_i$ for each $i$ and thus $T \mod p_i$ (since on the right hand
side of~\eqref{ccc} $T$ is multiplied by a term that has an inverse
modulo $p_i$). By the Chinese remainder theorem, we can compute $T$
in time polynomial in $n$ (see, e.\,g., ~\cite{MR1406794}, p.106).

\end{proof}

\subsection{{\sc \#ET} in 4-regular planar graphs is \#P-complete}

First, it's easy to see that {\sc \#ET} in 4-regular planar graphs is in
\#P. We will give a reduction from {\sc \#ET} in 4-regular graphs to {\sc \#ET}
in 4-regular planar graphs.

\begin{theorem}
{\sc \#ET} in 4-regular graphs is polynomial time Turing reducible
to {\sc \#ET} in 4-regular planar graphs.
\end{theorem}

\begin{proof}
Given a 4-regular graph $G=(V,E)$, we first draw $G$ in the plane.
We allow the edges to cross other edges, but i) edges do not cross
vertices, ii) each crossing involves $2$ edges. The
embedding can be found in polynomial time.

Let $p$ be an odd prime, we will construct a graph $G_p$ from the
embedded graph as follows. Let $e,e'$ be two edges in $G$ which
cross in the plane as shown in Figure~\ref{fig:cross}, we split $e$
(and $e'$) into two half-edges $e_1,e_2$ ($e'_1,e'_2$,
respectively). As illustrated in Figure~\ref{fig:gadget_cross}, a
$(0,X,Y)$ node with parameters $p$ and $k=p$ is added, and
$e_1,e'_1,e_2,e'_2$ are connected to the half-edges 0,1,2,3 of the
$(0,X,Y)$ node, respectively.

Let $G_p$ be the graph after replacing all
crossings by $(0,X,Y)$ nodes. We have that $G_p$ is planar
since $(X,Y,Y)$ nodes and $(0,X,Y)$ nodes are all planar. The
construction can be done in time polynomial in $p$ and the size of
$G$ (since the number of crossover points is at most $O(|E|^2)$ and
the size of each $(0,X,Y)$ node is $O(p^2)$).

\begin{figure}[htb]
\begin{center}
 \subfigure[Two crossing
 edges]{\label{fig:cross}\includegraphics[scale=0.5]{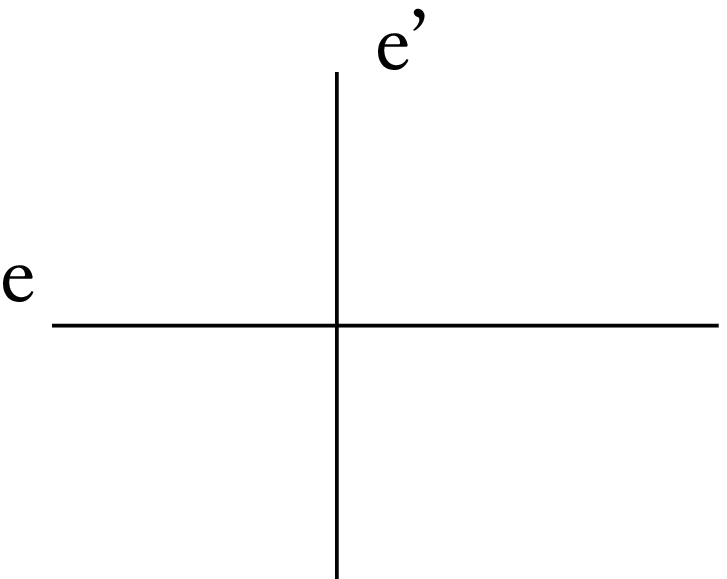}}\hspace{1cm}\quad
 \subfigure[After the replacement]{\label{fig:gadget_cross}\includegraphics[scale=0.4]{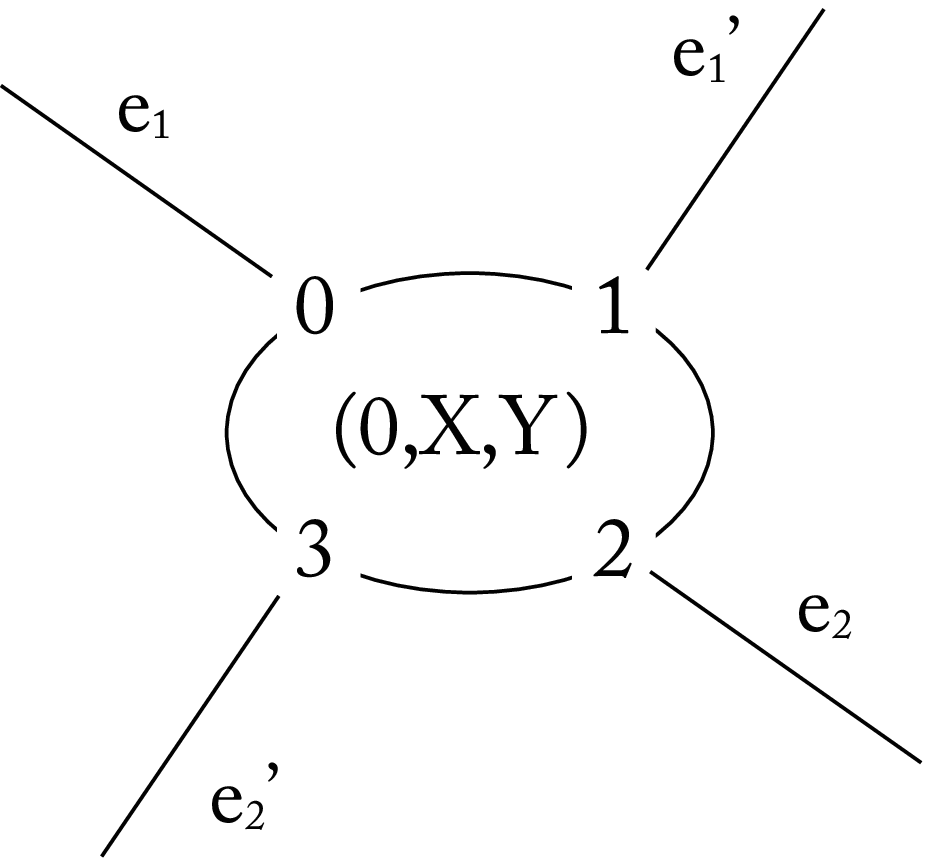}}
\end{center}
\caption{To replace a crossover point by a $(0,X,Y)$ node with
parameters $p$ and $k=p$}\label{fig:remove_cross}
\end{figure}

In the reduction, we choose $n = |V|$ primes $p_1,p_2,\ldots,p_n$
such that $p_i = O(n^2)$ for $i \in [n]$ and $\prod_{i=1}^n p_i \geq
3^n$, where $3^n$ is an upper bound for the number of Eulerian tours
in $G$ (the number of pairings in each vertex is $3$). For each
$p_i$, we construct a graph $G_{p_i}$ from the embedded graph as
described above with $p=p_i$. Let $T$ be the number of Eulerian
tours in $G$ and $T_i$ be the number of Eulerian tours in $G_{p_i}$,
we have
\begin{equation}\label{ddd}
T \equiv T_i \mod p_i.
\end{equation}
Equation~\eqref{ddd} follows from the fact that the number of
Eulerian tours in which the set of routes within any $(0,X,Y)$ node
is not of type $\{\{0,2\},\{1,3\}\}$ is zero (modulo $p_i$) (since
in~\eqref{ee2} we have $A\equiv 1\mod{p_i}$ and in~\eqref{ee3} we
have $B\equiv 0\mod{p_i}$). We can make a query to the oracle to
obtain the number $T_i$. By the Chinese remainder theorem, we can
compute $T$ in time polynomial in $n$. 
\end{proof}

\subsection{{\sc \#A-trails} in 4-regular graphs with rotational embedding schemes is \#P-complete}

In this section, we consider {\sc \#A-trails} in graphs with
rotational embedding schemes (maps). We prove that {\sc \#A-trails} in
$4$-regular maps is \#P-complete
by a simple reduction from {\sc \#ET} in 4-regular graphs.

First, it's not hard to verify that {\sc \#A-trails} in 4-regular
maps is in \#P.

\begin{theorem}
{\sc \#ET} in 4-regular graphs is polynomial time Turing reducible
to {\sc \#A-trails} in 4-regular maps.
\end{theorem}

\begin{figure}[htb]
\center
\includegraphics[scale=0.4]{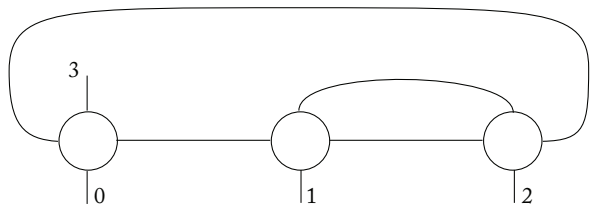}
\caption{Gadget simulating vertex of degree
4}\label{fig:gadget_4to4s}
\end{figure}

\begin{proof}
Given a 4-regular graph $G=(V,E)$, for each vertex $v$ of $G$, we
use the gadget shown in Figure~\ref{fig:gadget_4to4s} to replace
$v$.

The gadget consists of three vertices which are represented by
circles in Figure~\ref{fig:gadget_4to4s}. The labels 0, 1, 2 and 3
are the four half-edges which are used to replace half-edges
emanating from $v$. The cyclic ordering of the $4$ (half-)edges
incident to each circle is given by the clockwise order, as shown in
Figure~\ref{fig:gadget_4to4s}. There are three types of valid sets
of routes inside the gadget, ${\rm VR}(\{0, 1\},\{2, 3\})$, ${\rm VR}(\{0,
2\},\{1, 3\})$ and ${\rm VR}(\{0, 3\},\{1, 2\})$. By enumeration,
we have the size of each of the three sets is $2$.

Let $G'$ be the 4-regular map obtained by replacing each vertex $v$ by the gadget. Let $T$ be the
number of Eulerian tours in $G$, we have the number of A-trails in
$G'$ is $2^{|V|}T$. 
\end{proof}

Note that Kotzig~\cite{MR0248043} gave a one-to-one correspondence
between the A-trails in any 4-regular plane graph $G$ (the embedding
in the plane gives the rotational embedding scheme) and the spanning
trees in a plane graph $G'$, where $G$ is the medial graph of $G'$. By the Kirchhoff's theorem
(c.f.~\cite{MR1960003}), the number of spanning trees of any graph
can be computed in polynomial time. Thus {\sc \#A-trails} in
4-regular plane graphs can be computed in polynomial time.

\section{The complexity of approximate counting}\label{sec:app}

In this section, we show that {\sc \#ET} in general graphs is
AP-reducible to {\sc \#A-trails} in $4$-regular maps. AP-reductions were introduced by Dyer,
Goldberg, Greenhill and Jerrum~\cite{MR2044886} for the purpose of
comparing the complexity of two counting problems in terms of
approximation (given two counting problems $f,g$, if $f$ is
AP-reducible to $g$ and there is an FPRAS for $g$, then there is
also an FPRAS for $f$).

In the AP-reduction from {\sc \#ET} to {\sc \#A-trails} in
$4$-regular maps, we use the
idea of simulating the pairings in a vertex by a gadget as what we
did in the construction of the $Q$ gadget. The difference is that
the new gadget works in an approximate way, that is, instead of
having the number of valid sets of routes to be the same for each of
the types, the numbers can be different but within a small
multiplicative factor. The analysis of the gadget uses a fast mixing
result for a card shuffling problem.

We use the gadget illustrated
in Figure~\ref{fig:gadget_gto4s_det}. The circles represent the vertices in the map. Let $d$ be an even number.
The gadget has $d$ input half-edges on left and $d$ output
half-edges (Figure~\ref{fig:gadget_gto4s_det} demonstrates the case of $d=6$). There are $T$ layers in the
gadget which are numbered from 1 to $T$ from left to right. In an
odd layer $t$, the $(2i-1)$-st and the $2i$-th output half-edges
of layer $t-1$ are connected to a vertex of degree $4$, for
$i \in [d/2]$. In an even layer $t$, the $2i$-th and the
$(2i+1)$-st output half-edges of layer $t-1$ are connected to a
vertex of degree $4$, for $i \in [d/2-1]$. In Figure~\ref{fig:gadget_gto4s_det}, we illustrate the first
two layers each of which is in two consecutive vertical dashed lines. The
cyclic ordering of each vertex is given by the clockwise ordering
(in the drawing in Figure~\ref{fig:gadget_gto4s_det}), and so we
have that the two half-edges in each vertex which are connected to
half-edges of the previous layer are not cyclicly-adjacent.

\begin{figure}[htb]
\center
\includegraphics[scale=0.5]{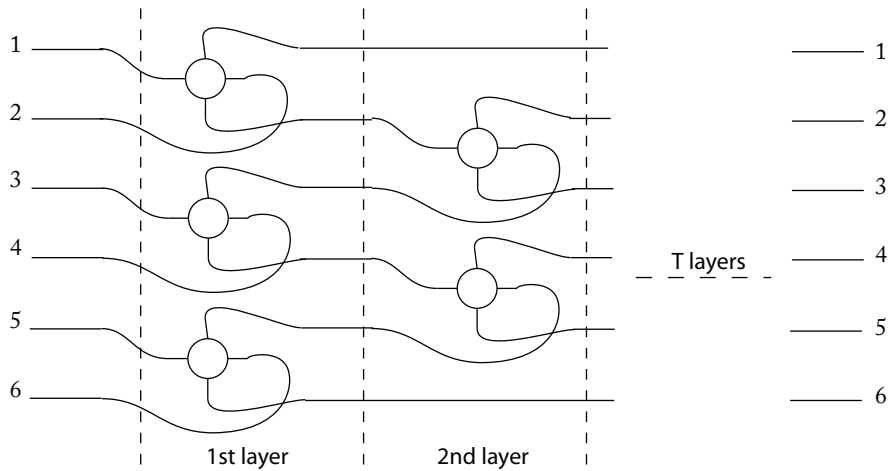}
\caption{Construction of the gadget for a vertex of degree
$6$}\label{fig:gadget_gto4s_det}
\end{figure}

Note that a valid route in the gadget always connects an input half-edge
to an output half-edge. Thus a valid set of routes always realizes
some permutation $\sigma$ connecting input half-edge $i$ to output
half-edge $\sigma(i)$.

In order to prove that $|{\rm VR}(\sigma)|$ is almost the same for each
permutation $\sigma \in S_d$, we show that for $T =
\Theta(d^2\log{d}\log(d!/\varepsilon))$ we have $$|{\rm VR}(\sigma)|/\sum_{\varrho
\in S_d}|{\rm VR}(\varrho)| \in [(1-\varepsilon)/d!,(1+\varepsilon)/d!]$$
for each permutation $\sigma \in S_d$. The gadget can be interpreted
as a process of a Markov chain for shuffling $d$ cards. The simplest
such chain proceeds by applying adjacent transpositions. The states
of the chain are all the permutations in $S_d$. In each time step,
let $\sigma \in S_d$ be the current state, we choose $i \in
\{1,\ldots,d-1\}$ uniformly at random, and then switch $\sigma(i)$
and $\sigma(i+1)$ with probability $1/2$ and stay the same with
probability $1/2$. For our gadget, it can be viewed as an even/odd
sweeping Markov chain on $d$ cards~\cite{MR2023023}. The ratio
$|{\rm VR}(\sigma)|/\sum_{\varrho \in S_d}|{\rm VR}(\varrho)|$ is exactly the
probability of being $\sigma$ at time $T$ when the initial state of
the even/odd sweeping Markov chain is the identity permutation. By
the analysis in~\cite{MR2023023}, we can relate $T$ with the ratio
as follows.

\begin{lemma}[\cite{MR2023023}]\label{ggg}
Let $T$ be the number of layers of the gadget with $d$ input
half-edges and $d$ output half-edges as shown in
Figure~\ref{fig:gadget_gto4s_det}, and let $\mu,\lambda$ be two
distributions on $S_d$ such that
$\mu(\sigma)=|{\rm VR}(\sigma)|/\sum_{\varrho \in S_d}|{\rm VR}(\varrho)|$ and
$\lambda(\sigma)=1/d!$ ($\lambda$ is the uniform distribution on
$S_d$). For
$$T=O(d^2\log {d}\log(d!/\varepsilon)),$$ then $\|\mu-\lambda\|_{TV}
\leq \varepsilon/d!$, and thus $(1-\varepsilon)/d! \leq \mu(\sigma)
\leq (1+\varepsilon)/d!$.
\end{lemma}

\begin{theorem}
If there is an FPRAS for {\sc \#A-trails} in 4-regular maps,
then we have an FPRAS for {\sc \#ET}
in general graphs.
\end{theorem}

\begin{proof}
Given an Eulerian graph $G=(V,E)$ and an error parameter $\varepsilon >
0$, we can, w.l.o.g., assume that the degree of vertices of $G$ is at
least $4$ (vertices of degree $2$ can be removed by contracting
edges). We construct graph $G'$ by replacing each vertex $v$ of
degree $d > 2$ with a gadget with $d$ input half-edges, $d$ output
half-edges and $T_d = \Theta(d^2\log{d}\log(4d!n/\varepsilon))$ layers where
the $(2i-1)$-st and $2i$-th output half-edge are connected (for $1
\leq i \leq d/2$), and the input half-edges are used to replace
half-edges emanating from $v$ (that is, they are connected to the
input half-edges of other gadgets according to the edge incidence at
$v$). We have that $G'$ has $O(n^2T_n)=O(n^4\log n(n\log
n+\log(1/\varepsilon)))$ vertices and can be constructed in time
polynomial in $n$ and $1/\varepsilon$.

Let $\mathcal{A}$ be an FPRAS for {\sc \#A-trails} in 4-regular
maps by the assumption of the
theorem, we run $\mathcal{A}$ on $G'$ with error parameter
$\varepsilon/2$. Let $\mathcal{A}(G',\varepsilon/2)$ be the output of
$\mathcal{A}$ and $N_A$ be the number of A-trails in $G'$, we have
$\mathcal{A}(G',\varepsilon/2) \in [e^{-\varepsilon/2}N_A,e^{\varepsilon/2}N_A]$
with probability at least 2/3. This process can be done in time
polynomial in the size of $G'$ and $1/\varepsilon$, which is polynomial
in $n$ and $1/\varepsilon$.

Let $D_d$ be the number of vertices in the gadget of $d$ input
half-edges and $d$ output half-edges, and let $R_d =
2^{D_d}2^{d/2}(d/2)!/d!$ and $R = \prod_{d=4}^n R_d^{n_d}$ where
$n_d$ is the number of vertices of degree $d$ in $G$. Our algorithm
$\mathcal{B}$ will output
\begin{equation}\label{fpras_out}
\mathcal{B}(G,\varepsilon)= \mathcal{A}(G',\varepsilon/2)/R.
\end{equation}

We next prove that $\mathcal{B}$ is an FPRAS for {\sc \#ET} in
general graphs. For every Eulerian tour in $G$, the type of the
pairing in each vertex in $G$ is fixed. Note that each pairing
corresponds to $(d/2)!2^{d/2}$ permutations in a gadget with $d$
input half-edges and $d$ output half-edges. By Lemma~\ref{ggg}, we
have $$(1-\varepsilon/(4n))2^{D_d}/d! \leq |{\rm VR}(\sigma)| \leq
(1+\varepsilon/(4n))2^{D_d}/d!$$ for each $\sigma \in S_d$ where
${\rm VR}(\sigma)$ is counted in a gadget with $d$ input half-edges and
$d$ output half-edges. Thus, the number of A-trails in $G'$ which
correspond to the same Eulerian tour in $G$ is in
$[(1-\varepsilon/(4n))^n R,(1+\varepsilon/(4n))^n R]$. Let $N_E$ be the number
of Eulerian tours in $G$, we have $$N_A \in [(1-\varepsilon/(4n))^n R
N_E,(1+\varepsilon/(4n))^n R N_E],$$ and thus for $\varepsilon \leq 2n$, $N_A/R
\in [e^{-\varepsilon/2} N_E, e^{\varepsilon/4} N_E]$ (the case when $\varepsilon
> 2n$ is trivial, $\mathcal{B}$ can just output $3^n$). Since $\mathcal{A}(G',\varepsilon/2) \in
[e^{-\varepsilon/2}N_A,e^{\varepsilon/2}N_A]$ with probability at least $2/3$,
then by~\eqref{fpras_out}, we have $\mathcal{B}(G,\varepsilon) \in
[e^{-\varepsilon}N_E,e^{\varepsilon}N_E]$ with probability at least $2/3$.
This completes the proof. 
\end{proof}

\section{The power of $4$-regular gadgets}\label{sec:signature}

In this section, we consider $4$-regular gadgets which are $4$-regular graphs (or maps) with $4$ half-edges (which are labeled from $0$ to $3$ and are the only connection from outside).
There are three types of valid sets of routes inside the gadget, ${\rm VR}(\{0, 1\}, \{2, 3\})$,
${\rm VR}(\{0, 2\}, \{1, 3\})$ and ${\rm VR}(\{0, 3\}, \{1, 2\})$. Since we are interested in the relative size
of the above three sets, we define the {\em signature} of a gadget to be a triple $(\alpha,\beta,\gamma)$ such that
\begin{eqnarray*}
\alpha &=& |{\rm VR}(\{0, 1\}, \{2, 3\})|/N,\\
\beta &=& |{\rm VR}(\{0, 2\}, \{1, 3\})|/N,\\
\gamma &=& |{\rm VR}(\{0, 3\}, \{1, 2\})|/N,
\end{eqnarray*}
where $N=|{\rm VR}(\{0, 1\}, \{2, 3\})|+|{\rm VR}(\{0, 2\}, \{1, 3\})|+|{\rm VR}(\{0, 3\}, \{1, 2\})|$. Note that $\alpha,\beta,\gamma \geq 0$ and $\alpha+\beta+\gamma=1$.

We will investigate what values of $(\alpha,\beta,\gamma)$ $4$-regular gadgets can achieve.
The motivation mainly comes from the question of whether a vertex in $4$-regular
maps can be simulated (exactly or approximately) by a $4$-regular graph gadget.
We think results in this section will give some insights on designing approximation
algorithms for \#{\sc ET} in $4$-regular graphs and \#{\sc A-trails} in $4$-regular maps.

We will discuss the power of
$4$-regular maps and $4$-regular graphs separately in the following subsections. Before that, we first
note that by permuting the labels of half-edges of a gadget of signature $(\alpha,\beta,\gamma)$, we have all the
permutations of $(\alpha,\beta,\gamma)$ as signatures. We next introduce an operation $2$-glue on gadgets
which constructs a new gadget.

\begin{figure}[htb]
\begin{center}
\includegraphics[scale=0.7]{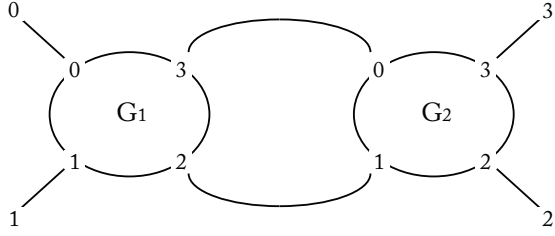}
\caption{$2$-glue of $G_1$ and $G_2$. The half-edges labeled $2$ and $3$ of $G_1$
are connected with half-edges labeled $1$ and $0$ of $G_2$, respectively. The
outmost straight lines are the half-edges of the $2$-glue of $G_1$ and $G_2$}\label{fig:2glue}
\end{center}
\end{figure}

Given two gadgets $G_1$ and $G_2$, the {\em $2$-glue} of $G_1$ and $G_2$ is a new gadget $G_3$ where half-edge $3$ and $2$ of $G_1$ are connected with half-edge $0$ and $1$ of $G_2$, respectively; and half-edge $0$ and $1$ of $G_1$ and half-edge $2$ and $3$ of $G_2$ are half-edges of $G_3$. The $2$-glue operation is illustrated in Figure~\ref{fig:2glue}.
Let $(\alpha_i,\beta_i,\gamma_i)$ be the signature of $G_i$, for $i=1,2,3$. By elementary counting, we have
\begin{eqnarray}
\alpha_3 &=& \frac{\alpha_1\beta_2+\alpha_1\gamma_2+\beta_1\alpha_2+\gamma_1\alpha_2}{1-\alpha_1\alpha_2},\label{eq:new1}\\
\beta_3 &=& \frac{\beta_1\gamma_2+\gamma_1\beta_2}{1-\alpha_1\alpha_2},\label{eq:new2}\\
\gamma_3 &=& \frac{\beta_1\beta_2+\gamma_1\gamma_2}{1-\alpha_1\alpha_2}.\label{eq:new3}
\end{eqnarray}

\subsection{The power of $4$-regular maps}

We will show in this section that $4$-regular map gadgets can achieve almost all rational
points $(\alpha,\beta,\gamma)$ on the plane $\alpha+\beta+\gamma=1$ and $\alpha,\beta,\gamma \geq 0$.

\begin{theorem}\label{thm:mg_pos}
For every $\alpha,\beta,\gamma \in \Q$ such that $0 \leq \alpha,\beta,\gamma < 1$ and $\alpha+\beta+\gamma=1$, there is a $4$-regular map gadget having signature $(\alpha,\beta,\gamma)$.
\end{theorem}

The simplest $4$-regular map gadget $SMG$ is one vertex with $4$ half-edges. The signature of $SMG$ is $(1/2,1/2,0)$.
We will prove Theorem~\ref{thm:mg_pos} by showing that starting from $SMG$
and by applying the $2$-glue operation, we can achieve almost all rational signatures.

\begin{lemma}\label{lem:m1}
For every $q \in \Q^+$, there is a $4$-regular map gadget with signature $(\alpha,\beta,0)$ such that
$$\alpha/\beta = q.$$
\end{lemma}

\begin{proof}
Given a gadget with signature $(\alpha_1,\beta_1,0)$, we can permute the labels of the half-edges (without changing the size of the gadget) to achieve signature $(\beta_1,\alpha_1,0)$. Let $q_1=\alpha_1/\beta_1$. The above operation on gadgets defines a mapping
\begin{equation}\label{eq:map1}
q_1 \mapsto 1/q_1.
\end{equation}

Given two gadgets $G_1$ and $G_2$ with signatures $(\alpha_1,\beta_1,0)$ and $(\alpha_2,0,\gamma_2)$, respectively, let $G_3$ be the $2$-glue of $G_1$ and $G_2$. By~\eqref{eq:new1}--\eqref{eq:new3}, the signature of $G_3$ is
\begin{equation}\label{eq:ratiosum}
(\alpha_3,\beta_3,\gamma_3)=\left(\frac{\alpha_1\gamma_2+\beta_1\alpha_2}{1-\alpha_1\alpha_2},\frac{\beta_1\gamma_2}{1-\alpha_1\alpha_2},0\right).
\end{equation}
Let $q_1 = \alpha_1/\beta_1,q_2=\alpha_2/\gamma_2,q_3=\alpha_3/\beta_3$, by~\eqref{eq:ratiosum} we have $q_3=q_1+q_2$. The above operation defines a mapping \begin{equation}\label{eq:map2}
(q_1,q_2) \mapsto q_1+q_2.
\end{equation}

To prove the lemma, it is sufficient to prove that starting from $1$, there is a sequence of mappings using~\eqref{eq:map1} and~\eqref{eq:map2} which
will achieve any $q \in \Q^+$. Assume $q=s/t$ and $s,t \in \N$. Starting from $1$, we can achieve $t$ by using a sequence of~\eqref{eq:map2}. Then by using~\eqref{eq:map1} on $t$, we can achieve $1/t$. Finally, we can achieve $s/t$ by using a sequence of~\eqref{eq:map2} on $1/t$. This completes the proof.
\end{proof}

\begin{proof}[Proof of Theorem~\ref{thm:mg_pos}]
For every $0 \leq \alpha,\beta,\gamma <1$ and $\alpha+\beta+\gamma=1$, if one of $\alpha,\beta,\gamma$ is zero, then we are done by Lemma~\ref{lem:m1}. We next assume that $0 < \alpha,\beta,\gamma < 1$.
By Lemma~\ref{lem:m1}, there are gadgets $G_1$ and $G_2$ with signatures $(0,\beta/(1-\alpha),\gamma/(1-\alpha))$ and $(\alpha,0,1-\alpha)$, respectively. Let $G$ be the $2$-glue of $G_1$ and $G_2$, by~\eqref{eq:new1}--\eqref{eq:new3}, the signature of $G$ is $(\alpha,\beta,\gamma)$.
\end{proof}

\subsection{The power of $4$-regular graphs}

In this section, we investigate the power of $4$-regular graph gadgets. Let $S$ contains of vectors
$$
P (\alpha,\beta,\gamma)^{\mathrm T},
$$
where $P$ is a $3 \times 3$ permutation matrix, $\alpha \geq \beta \geq \gamma \geq 0$, $\alpha+\beta+\gamma=1$ and
$\gamma \geq f(\beta)$, where
\begin{equation}\label{eq:boundary}
f(x) = \frac{1}{2}(1-x)(1-\exp(2x/(x-1))).
\end{equation}
$S$ is illustrated in Figure~\ref{fig:boundary}. We first show that signatures in $S$ are achievable by $4$-regular graph gadgets:

\begin{figure}[htb]
\begin{center}
\includegraphics[scale=0.3]{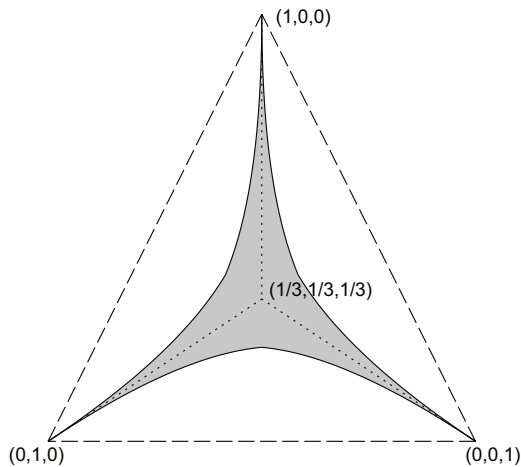}
\caption{The shaded region $S$ contains the signatures which can be achieved by $4$-regular graph gadgets as described in Theorem~\ref{thm:gg_pos}. The triangle region within dashed lines contains signatures $(\alpha,\beta,\gamma)$ s.t. $\alpha+\beta+\gamma=1$ and $\alpha,\beta,\gamma \geq 0$.}\label{fig:boundary}
\end{center}
\end{figure}

\begin{theorem}\label{thm:gg_pos}
For every $\vec{s} \in S$, and for every $\varepsilon > 0$, there is a $4$-regular graph gadget with signature $\vec{s}'$ such that
$$
||\vec{s}-\vec{s}'||_1 \leq \varepsilon.
$$
\end{theorem}

On the other hand, we will show $S$ is closed under $2$-glue operations.

\begin{theorem}\label{thm:gg_neg}
For $i=1,2$, let $G_i$ be a $4$-regular graph gadget with signature $(\alpha_i,\beta_i,\gamma_i) \in S$. Let $G_3$ be the $2$-glue of $G_1$ and $G_2$ with signature $(\alpha_3,\beta_3,\gamma_3)$ defined by~\eqref{eq:new1}--\eqref{eq:new3}, hence $(\alpha_3,\beta_3,\gamma_3) \in S$.
\end{theorem}

We performed experiment on all gadgets up to $7$ vertices with random signature from $S$ for each vertex, the result was in $S$. It seems that $S$ is the biggest region we can get for the signatures of $4$-regular graph gadgets. Based on the results of our experiment, we conjecture that $S$ contains all signatures of $4$-regular graph gadgets.

\begin{conjecture}
For every $4$-regular graph gadget with signature $\vec{s}$, $\vec{s} \in S$.
\end{conjecture}

In the rest of this section, we will prove Theorem~\ref{thm:gg_pos} and Theorem~\ref{thm:gg_neg}.

\subsubsection{Proof of Theorem~\ref{thm:gg_pos}}

The simplest $4$-regular graph gadget $SGG$ is one vertex with $4$ half-edges, and the signature of $SGG$ is $(1/3,1/3,1/3)$. We will first show in the following lemma, that starting from $SGG$, we can construct a gadget with signature close enough to $(\alpha,\alpha,1-2\alpha)$, for $\alpha \in (0,1/3)$. This achieves all points on the segments $((1,0,0),(1/3,1/3,1/3))$, $((0,1,0),(1/3,1/3,1/3))$, and $((0,0,1),(1/3,1/3,1/3))$, shown as dotted lines in Figure~\ref{fig:boundary}.

\begin{lemma}\label{lem:gg_pos1}
For every $r,\varepsilon \in \R$, $0 \leq r \leq 1$ and $\varepsilon > 0$, there is a gadget $G$ (a $4$-regular graph with $4$ half-edges) of size $n$ and with signature $(\alpha,\alpha,1-2\alpha)$ such that $|r-\alpha/(1-2\alpha)| \leq \varepsilon$ and $n \leq \lceil 2/\varepsilon \rceil$.
\end{lemma}

\begin{proof}
Let $G_1$ be a gadget with signature $(\alpha_1,\alpha_1,1-2\alpha_1)$. By~\eqref{eq:new1}--\eqref{eq:new3}, the $2$-glue of $G_1$ and $SGG$ is
\begin{equation*}
(\alpha,\beta,\gamma)=\left(\frac{1+\alpha_1}{3-\alpha_1},\frac{1-\alpha_1}{3-\alpha_1},\frac{1-\alpha_1}{3-\alpha_1}\right).
\end{equation*}
Let $q_1 = \alpha_1/(1-2\alpha_1)$, $q=\beta/\alpha=(1-\alpha_1)/(1+\alpha_1)$. Hence we have $q=(1+q_1)/(1+3q_1)$. In this way, we define a mapping
\begin{equation}\label{eq:mapt1}
q_1 \mapsto (1+q_1)/(1+3q_1).
\end{equation}

Let $G_2$ be a gadget with signature $(1-2\alpha_2,\alpha_2,\alpha_2)$. . By~\eqref{eq:new1}--\eqref{eq:new3}, the $2$-glue of $G_2$ and $SGG$ is
\begin{equation*}
(\alpha,\beta,\gamma)=\left(\frac{1-\alpha_2}{1+\alpha_2},\frac{\alpha_2}{1+\alpha_2},\frac{\alpha_2}{1+\alpha_2}\right).
\end{equation*}
Let $q_2 = \alpha_2/(1-2\alpha_2)$, $q=\beta/\alpha=\alpha_2/(1-\alpha_2)$. Hence we have $q=q_2/(1+q_2)$. In this way, we define another mapping
\begin{equation}\label{eq:mapt2}
q_2 \mapsto q_2/(1+q_2).
\end{equation}

To prove the statement of the lemma, it is sufficient to prove that starting from $1$ and by applying mappings~\eqref{eq:mapt1} and~\eqref{eq:mapt2}, we can achieve $q$ which is close enough to $r$.

\begin{claim}
Let $L_1=U_1=W_1=\emptyset$. Define for every $i > 1$
$$L_{i+1} = \{ q/(1+q) \,|\, q \in W_{i}\cup\{1\} \},$$
$$U_{i+1} = \{ (1+q)/(1+3q) \,|\, q \in W_{i}\cup\{1\} \},$$
$$W_{i+1} = L_{i+1} \cup U_{i+1}.$$
For $i \geq 3$, the following properties hold:
\begin{enumerate}
\item the minimum element in $L_i$ is $1/i$, the maximum element in $L_i$ is $1/2$;
\item let $q_1,q_2$ be consecutive elements in $L_i$, (that is, there is no $q_3 \in L_i$
satisfying $q_1 < q_3 < q_2$,) then $q_2-q_1 \leq 1/i$;
\item the maximum element in $U_i$ is $i/(i+2)$, the minimum element in $U_i$ is $1/2$;
\item let $q_1,q_2$ be consecutive elements in $U_i$, then $q_2-q_1 \leq 2/(i+2)$.
\end{enumerate}
\end{claim}

\begin{proof}[Proof of Claim]
We prove the claim by induction on $i$. For the base case of $i=3$, $L_3=\{1/3,1/2\}$ and $U_3=\{1/2,3/5\}$, the statement is true.

We next consider the case of $i = k>3$. Note that the mapping $q \mapsto q/(1+q)$ is monotonically increasing. By the induction hypothesis (the first and the third property),
\begin{equation}\label{eq:hypo1}
\begin{split}
\mbox{the minimum element in $W_{k-1}$ is $1/(k-1)$,}\\
\mbox{and the maximum element in $W_{k-1}$ is $(k-1)/(k+1)$.}
\end{split}
\end{equation}
Hence, the minimum element in $L_k$ is $1/k$ and the maximum element is $1/2$ (which proves the first property).

We next prove the second property, that is, consecutive elements $q'_1,q'_2 \in L_k$ satisfy $q'_2-q'_1 \leq 1/k$. There are two cases depending on whether $q'_1,q'_2$ are both mapped from $L_{k-1}$ or from $U_{k-1}$. (The case that one is mapped from $L_{k-1}$ and the other from $U_{k-1}$ can happen only if one of $q'_1,q'_2$ is mapped from $1/2$. This is because the mapping $q \mapsto q/(1+q)$ is monotone and $1/2$ is the only element in both $L_{k-1}$ and $U_{k-1}$ by induction hypothesis.)
\begin{description}
\item[Case 1: both from $L_{k-1}$.] Let $q_1,q_2$ be consecutive elements in $L_{k-1}$ such that $q'_i = q_i/(1+q_i)$, for $i=1,2$. Then by the induction hypothesis (the first and the second property), we have
\begin{equation}\label{eq:hypo2}
\begin{split}
q_1,q_2 \geq 1/(k-1),\\
q_2-q_1 \leq 1/(k-1).
\end{split}
\end{equation}
Then we have
\begin{eqnarray*}
q'_2 - q'_1 & = & \frac{q_2-q_1}{(1+q_1)(1+q_2)} \leq \frac{1}{k-1} \cdot \frac{1}{(1+1/(k-1))(1+1/(k-1))} \leq \frac{1}{k},
\end{eqnarray*}
where the second inequality follows from~\eqref{eq:hypo2}.
\item[Case 2: both from $U_{k-1}$.] Let $p_1,p_2$ be consecutive elements in $U_{k-1}$ such that  $q'_i = p_i/(1+p_i)$, for $i=1,2$. Then by the induction hypothesis (the first, the third and the fourth property), we have
\begin{equation}\label{eq:hypo3}
\begin{split}
p_1,p_2 \geq 1/2 \geq 1/(k-1),\\
p_2-p_1 \leq 2/(k+1).
\end{split}
\end{equation}
Then we have
\begin{eqnarray*}
q'_2 - q'_1 & = & \frac{p_2-p_1}{(1+p_1)(1+p_2)} \leq \frac{2}{k+1} \cdot \frac{1}{(1+1/2)(1+1/2)} <  \frac{1}{k},
\end{eqnarray*}
where the second inequality follows from~\eqref{eq:hypo3}.
\end{description}

Note that the mapping $q \mapsto (1+q)/(1+3q)$ is monotonically decreasing. By~\eqref{eq:hypo1}, we have
the maximum element in $U_k$ is $k/(k+2)$ and the minimum element is $1/2$, (which proves the third property).

We next prove the fourth property, that is, consecutive elements $q'_1,q'_2 \in U_k$ satisfy $q'_2-q'_1 \leq 2/(k+2)$. There are two cases depending on whether $q'_1,q'_2$ are both mapped from $L_{k-1}$ or from $U_{k-1}$.
\begin{description}
\item[Case 1: both from $L_{k-1}$.] Let $q_1,q_2$ be consecutive elements in $L_{k-1}$ such that $q'_i = (1+q_i)/(1+3q_i)$, for $i=1,2$. Then we have
\begin{eqnarray*}
q'_2 - q'_1 & = & \frac{2(q_2-q_1)}{(1+3q_1)(1+3q_2)} \leq \frac{2}{k-1} \cdot \frac{1}{(1+3/(k-1))(1+3/(k-1))} \leq \frac{2}{k+2},
\end{eqnarray*}
where the second inequality follows from~\eqref{eq:hypo2}.
\item[Case 2: both from $U_{k-1}$.]
Let $p_1,p_2$ be any two distinct elements in $U_{k-1}$ such that $q'_i = (1+p_i)/(1+3p_i)$, for $i=1,2$, then we have
\begin{eqnarray*}
q'_2 - q'_1 & = & \frac{2(p_2-p_1)}{(1+3p_1)(1+3p_2)} \leq  \frac{2}{k+1} \cdot \frac{2}{(1+3/(k-1))(1+3/2)} <  \frac{2}{k+2},
\end{eqnarray*}
where the second inequality follows from~\eqref{eq:hypo3}.
\end{description}

\end{proof}

The lemma then follows from the claim. This is because $W_{i} \cup \{1\}$ is the set of $q$ we can achieve after we applying a sequence of $i-1$ mappings from $1$ using~\eqref{eq:mapt1} and~\eqref{eq:mapt2}, and the difference between two consecutive elements in $W_{i} \cup \{1\}$ is at most $2/(i+2)$.

\end{proof}

\begin{proof}[Proof of Theorem~\ref{thm:gg_pos}]
Given any signature $\vec{s} \in S$ for which we want to construct a gadget, w.l.o.g., we assume that $\vec{s}_3 \geq \vec{s}_1 \geq \vec{s}_2$.
Let $G$ be a gadget with signature $(\delta/2,1-\delta,\delta/2)$, for some $\delta \in (0,1)$ to be fixed later. Let $G_0$ be a gadget with signature $(\alpha_0,\alpha_0,1-2\alpha_0)$, where $\alpha_0$ depends on $\vec{s}$ and $\delta$, and will be fixed later. We define a sequence of gadgets $G_1,\ldots,G_t$ such that $G_i$ is a $2$-glue of $G_{i-1}$ and $G$, for $1 \leq i \leq t$. We will choose $t$ such that $\alpha_{t-1} < \vec{s}_1 \leq \alpha_t$ where $(\alpha_{i},\beta_{i},\gamma_{i})$ is the signature of $G_i$, for $1 \leq i \leq t$. Our goal is to show that by properly choosing $\delta$, $(\alpha_t,\beta_t,\gamma_t)$ is close to $\vec{s}$.

Now we show how to choose the initial $G_0$. For $c \in [0,1]$, we define a function $f_c(x)$ as follows.
\begin{equation*}
f_c(x) = -\frac12(x-1)(1+(2c-1)\exp(2x/(x-1))).
\end{equation*}
We have
\begin{equation}\label{eq:ggp8}
\frac{{\mathrm d}}{{\mathrm d}x}f_c(x) = \frac{-3f_c(x)+1-x+f_c(x)x}{(x-1)^2}.
\end{equation}
Note that when $x \in [0,1)$, $f_c(x)$ is concave for every $c \in [0,1/2]$. This is because the second derivative of $f_c(x)$ is
$$
2(1-2c)\exp(2x/(x-1))/(x-1)^3,
$$
which is non-positive when $x \in [0,1)$ and $c \in [0,1/2]$. Also note that for every $x \in [0,1)$ and $c > c'$, 
\begin{equation}\label{eq:01}
f_{c}(x) \geq f_{c'}(x).
\end{equation}

Since $\vec{s} \in S$ and $\vec{s}_3 \geq \vec{s}_1 \geq \vec{s}_2$, by the definition of $S$, we have $\vec{s}_2 \geq f(\vec{s}_1)$ where $f(x)$ is defined in~\eqref{eq:boundary}. Note that $f(x)=f_0(x)$. Thus we have
\begin{equation}\label{eq:ggp5}
\vec{s}_2 \geq f_0(\vec{s}_1).
\end{equation}
Since $\vec{s}_2 \leq \vec{s}_3 = 1-\vec{s}_1-\vec{s}_2$, then we have
\begin{equation}\label{eq:ggp6}
\vec{s}_2 \leq (1-\vec{s}_1)/2 = f_{1/2}(\vec{s}_1).
\end{equation}
By \eqref{eq:ggp5}, \eqref{eq:ggp6} and continuity, there is $c^* \in [0,1/2]$ such that $\vec{s}_2 = f_{c^*}(\vec{s}_1)$. 

Now we show that $f_{c^*}(x)=x$ has at least one solution in $[0,1/3]$. This follows from the facts that $f_{c^*}(0)=c^* \geq 0$, $f_{c^*}(1/3) \leq f_{1/2}(1/3) = 1/3$, and $f_{c^*}(x)$ is continuous. We let $\alpha$ be a solution of $f_{c^*}(x)=x$, that is
\begin{equation}\label{eq:ggp9}
f_{c^*}(\alpha)=\alpha.
\end{equation}
By Lemma~\ref{lem:gg_pos1}, there is a gadget with signature $(\alpha_0,\alpha_0,1-2\alpha_0)$ such that $0 \leq \alpha_0-\alpha \leq \delta$. This will be our $G_0$.

We use the following method to estimate the signature of $G_i$.
For $0 \leq i \leq t-1$, $G_{i+1}$ has signature
\begin{equation*}
(\alpha_{i+1},\beta_{i+1},\gamma_{i+1})=\Big(\frac{2\alpha_i-\alpha_i\delta+\beta_i\delta+\gamma_i\delta}{2-\alpha_i\delta}, \frac{2\beta_i-2\beta_i\delta+\gamma_i\delta}{2-\alpha_i\delta}, \frac{\beta_i\delta+2\gamma_i-2\gamma_i\delta}{2-\alpha_i\delta}\Big).
\end{equation*}
Note that $\alpha_{i+1} \geq \alpha_i$, for $0 \leq i \leq t-1$.

By taking the derivatives of $(\alpha_{i+1},\beta_{i+1},\gamma_{i+1})$ w.r.t. $\delta$, we have
\begin{eqnarray}
\frac{\partial \alpha_{i+1}}{\partial \delta}\Big|_{\delta=0} &=& -\alpha_i+1/2+\alpha_i^2/2 = (\alpha_i-1)^2/2,\label{eq:ggp1}\\
\frac{\partial \beta_{i+1}}{\partial \delta}\Big|_{\delta=0} &=& -3\beta_i/2+1/2-\alpha_i/2+\alpha_i\beta_i/2,\label{eq:ggp2}\\
\frac{\partial \gamma_{i+1}}{\partial \delta}\Big|_{\delta=0} &=& 3\beta_i/2-1+\alpha_i+(1-\alpha_i-\beta_i)\alpha_i/2.\nonumber
\end{eqnarray}
Thus, by Taylor expansion, we have
\begin{equation*}
(\alpha_{i+1},\beta_{i+1},\gamma_{i+1})=(\alpha_i,\beta_i,\gamma_i) + \Big(\frac{\partial \alpha_{i+1}}{\partial \delta}\Big|_{\delta=0},\frac{\partial \beta_{i+1}}{\partial \delta}\Big|_{\delta=0},\frac{\partial \gamma_{i+1}}{\partial \delta}\Big|_{\delta=0}\Big)\delta + O(\delta^2).
\end{equation*}
Note that
\begin{equation}\label{eq:ggp10}
\alpha_{i+1} - \alpha_i = \Theta(\delta),
\end{equation}
by~\eqref{eq:ggp1} and the fact that $0 \leq \alpha_{t-1} < \vec{s}_1 \leq 1/2$.
If $\beta_i=f_{c_i}(\alpha_i)$ for some $c_i \in [0,1]$, then by~\eqref{eq:ggp1}, \eqref{eq:ggp2} and~\eqref{eq:ggp8}, we have
\begin{equation}\label{eq:ggp4}
|\beta_{i+1}-f_{c_i}(\alpha_{i+1})| \leq O(\delta^2).
\end{equation}

For the sequence of $(\alpha_i,\beta_i,\gamma_i)_{i=0}^t$, we will show that each $(\alpha_i,\beta_i)$ is on the curve $f_{c_i}$ for some $c_i \in [0,1/2]$ and each $f_{c_i}$ will not deviate from $f_{c^*}$ by too much. In this way, we can upper bound the distance between $(\alpha_t,\beta_t,\gamma_t)$ and~$\vec{s}$.

First we show, by induction, that for small enough $\delta$, for every $0 \leq i \leq t$, there is some $c_i \in [0,1/2]$ such that $\beta_i = f_{c_i}(\alpha_i)$. We consider the case of $i=0$. By~\eqref{eq:ggp8} and~\eqref{eq:ggp9}, we have the derivative of $f_{c^*}(x)$ at $x=\alpha$ is less than $1$. Since $\alpha_0 \geq \alpha$, by~\eqref{eq:01} and the concavity of $f_{c^*}(x)$, we have $\alpha_0 = f_{c_0}(\alpha_0)$ for some $c_0 \geq c^*$. Moreover, since $\alpha_0 \leq \alpha+\delta$, we have $c_0 \leq 1/2$ for some small enough $\delta$.

Now we assume that there is $c_k \in [0,1/2]$ such that $f_{c_k}(\alpha_k)=\beta_k$. Note that by~\eqref{eq:ggp1}, $\alpha_{k+1} \geq \alpha_k$ for small enough $\delta$. Since $f_{c_k}(x)$ is concave, by~\eqref{eq:ggp1}, \eqref{eq:ggp2}, \eqref{eq:ggp8} and~\eqref{eq:01}, there is $c_{k+1} \geq c_k$ such that $f_{c_{k+1}}(\alpha_{k+1})=\beta_{k+1}$. Moreover, note that if $\beta_\ell =f_{1/2}(\alpha_\ell)$ then $\beta_{\ell+1} =f_{1/2}(\alpha_{\ell+1})$ for every $0 \leq \ell \leq t-1$. Hence for small enough $\delta$, we have $c_{k+1} \leq 1/2$. Thus, there is $c_i \in [0,1/2]$ such that $\beta_i = f_{c_i}(\alpha_i)$ for all $0 \leq i \leq t$.

We next claim that for every $c > c'$, $f_c(x)-f_{c'}(x)$ is monotonically decreasing. This is because
$$
\frac{{\mathrm d}}{{\mathrm d}x}(f_c(x)-f_{c'}(x)) = -\frac{(x-3)(c-c')}{x-1}\exp(2x/(x-1)),
$$
which is negative when $x \in [0,1)$.

Now we prove $|\beta_t-\vec{s}_2| \leq O(\delta)$. By~\eqref{eq:ggp10}, and the fact that $\vec{s}_1 \leq 1/2$, we have
\begin{equation}\label{eq:ggp7}
t=O(1/\delta).
\end{equation}
Hence, by~\eqref{eq:ggp4},~\eqref{eq:ggp7} and the above claims, we have
\begin{eqnarray*}
|\beta_t - \vec{s}_2| &\leq& |f_{c_0}(\alpha_t)-\vec{s}_2| + \sum_{i=0}^{t-1} |f_{c_{i+1}}(\alpha_t)-f_{c_i}(\alpha_t)|\\
&\leq& |f_{c_0}(\alpha_t)-f_{c^*}(\alpha_t)| + |f_{c^*}(\alpha_t)-\vec{s}_2| + \sum_{i=0}^{t-1} |f_{c_{i+1}}(\alpha_{i+1})-f_{c_i}(\alpha_{i+1})|\\
&\leq& |f_{c_0}(\alpha_0)-f_{c^*}(\alpha_0)| + |f_{c^*}(\alpha_t)-\vec{s}_2| + \sum_{i=0}^{t-1} |\beta_{i+1}-f_{c_i}(\alpha_{i+1})|\\
&\leq& O(\delta).
\end{eqnarray*}

Since $|\alpha_t-\vec{s}_1| \leq O(\delta)$, and $|\gamma_t-\vec{s}_3| \leq |\alpha_t-\vec{s}_1|+|\beta_t-\vec{s}_2|$, we can choose $\delta = O(\varepsilon)$ to make
$$
|\alpha_t-\vec{s}_1|+|\beta_t-\vec{s}_2|+|\gamma_t-\vec{s}_3| \leq \varepsilon.
$$
This completes the proof.

\end{proof}

\subsubsection{Proof of Theorem~\ref{thm:gg_neg}}
Let $u$ be the solution of
$$
\frac{1}{2}(1-x)(1+\exp(2x/(x-1))) = x,
$$
on the interval $[0,1]$. Since the left hand side is decreasing in $x$, we have that the solution is unique. We have $u \leq 39/100$.

From~\eqref{eq:boundary}, it follows that the boundary $\partial S$ of $S$ contains the signatures
$$
P\left(\frac{1}{2}(1-y)(1+\exp(2y/(y-1))),y,\frac{1}{2}(1-y)(1-\exp(2y/(y-1)))\right)^{\mathrm T},
$$
where $P$ is a $3 \times 3$ permutation matrix, and $y \in [0,u]$. For $y=x/(x+1)$, we obtain that $\partial S$
contains the signatures
$$
P\left(\frac{\cosh(x)}{(x+1){\mathrm e}^x},\frac{x{\mathrm e}^x}{(x+1){\mathrm e}^x},\frac{\sinh(x)}{(x+1){\mathrm e}^x}\right)^{\mathrm T},
$$
where $x \in [0,u/(1-u)]$. Let $w=u/(1-u)$, we have $1/2 \leq w\leq 64/100$.

We next show that the $2$-glue of two signatures in $\partial S$ results in a signature in $S$.
In the proof, we may use the following bounds for ${\mathrm e}^z$:
\begin{gather}
1+z \leq {\mathrm e}^z, \quad \forall z\in\R;\label{eq:exp1}\\
1+z+ \frac{z^2}{2}\leq{\mathrm e}^z\leq 1+z+z^2, \quad \forall z \in [0,2w];\label{eq:exp2}\\
1+z+\frac{z^2}{2}+\frac{z^3}{6}\leq{\mathrm e}^z\leq 1+z+\frac{z^2}{2}+\frac{z^3}{3}, \quad \forall z \in [0,2w];\label{eq:exp3-1}\\
1+z+\frac{z^2}{2}+\frac{z^3}{6}\leq{\mathrm e}^z\leq 1+z+\frac{z^2}{2}+\frac{2z^3}{5}, \quad \forall z \in [0,4w];\label{eq:exp3-2}\\
1+z+\frac{z^2}{2}+\frac{z^3}{6}+\frac{z^4}{24}\leq{\mathrm e}^z\leq 1+z+\frac{z^2}{2}+\frac{z^3}{6}+\frac{z^4}{12}, \quad \forall z \in [0,3w];\label{eq:exp4}\\
\begin{split}
1+z+\frac{z^2}{2}+\frac{z^3}{6}+\frac{z^4}{24}+\frac{z^5}{120}\leq{\mathrm e}^z\leq1+z+\frac{z^2}{2}+\frac{z^3}{6}+\frac{z^4}{24}+\frac{z^5}{60}, \quad \forall z \in [0,4w].\label{eq:exp5}
\end{split}
\end{gather}

We also need the following two lemmas.

\begin{lemma}\label{lem:normal}
Let
$$
\vec{n}=\frac{\mu}{2} P(-{\mathrm e}^{2x}+1+2x, -2, {\mathrm e}^{2x}+2x+1)^{\mathrm{T}},
$$
be the normal of the surface
$$
\alpha P(\cosh(t),t{\mathrm e}^t,\sinh(t))^{\mathrm{T}}
$$
at point $\alpha=\mu$ and $t=x$, where $\mu \geq 0$, $P$ is a $3 \times 3$ permutation matrix, and $x \in [0,w]$.
For every $x,y \in [0,w]$, we have
$$
\vec{n}^{\mathrm{T}} P (\cosh(y),y{\mathrm e}^y,\sinh(y))^{\rm T} \leq 0.
$$
\end{lemma}

\begin{proof}
Using~\eqref{eq:exp1} we obtain
$$
{\mathrm e}^{2x-y} \geq {\mathrm e}^y (1+2x-2y)
$$
and hence
\begin{equation*}
(-{\mathrm e}^{2x}+1+2x,-2,{\mathrm e}^{2x}+1+2x) (\cosh(y),y{\mathrm e}^y,\sinh(y))^{\rm T}=-\mathrm{e}^{2x-y} + {\mathrm e}^y (1 + 2x - 2 y)\leq 0.
\end{equation*}

\end{proof}

\begin{lemma}\label{lem:jacob}
Let $D$ be simply connected and compact subset of $\R^2$. Let
$g:D \to \R^2$ be a continuously differentiable map. Assume that
\begin{itemize}
\item the image of the boundary (that is, $g(\partial D)$) is a simple
curve, and
\item the Jacobian determinant of $g$ does not vanish in the interior of $D$.
\end{itemize}
Then $g(D)$ is contained inside of $g(\partial D)$.
\end{lemma}

\begin{proof}
It follows from the Jordan curve theorem that the inside and the
outside of $g(\partial D)$ are well-defined.
For the sake of contradiction, suppose that $x$ from the interior of
$D$ gets mapped outside of $g(\partial D)$. Let $C$ be a smooth simple
curve from $g(x)$ to infinity such that $C$ does not intersect
$g(\partial D)$. The image $g(D)$ is a compact
set and hence $C$ intersects $\partial g(D)$. Thus there exists a
point $y$ such that $y\in \partial g(D)$
and $y\not\in g(\partial D)$. Let $z$ be the preimage of $y$ in the
interior of $D$. By our assumption the
Jacobian determinant of $g$ does not vanish at $z$ and hence $y$ is
regular. Since $z$ is in the interior of $D$
we have that a small neighborhood of $z$ is in $D$ and hence (using
regularity of $y$) a small neighborhood
of $y=g(z)$ is in $g(D)$, a contradiction with $y\in\partial g(D)$.

\end{proof}

\begin{lemma}\label{lem:boundaries}
Let $G_1,G_2$ be two gadgets with signatures $\vec{s}_1,\vec{s}_2 \in \partial S$, respectively, then the signature of the $2$-glue of $G_1$ and $G_2$ is in $S$.
\end{lemma}

\begin{proof}
Let $T$ be the cone containing the vectors $\mu \vec{s}$, where $\mu$ is a non-negative real number and $\vec{s} \in S$. The boundary $\partial T$ of $T$ contains the vectors
$$
\mu P (\cosh(x),x{\mathrm e}^x,\sinh(x))^{\mathrm T},
$$
where $P$ is a $3 \times 3$ permutation matrix, and $x \in [0,w]$.

We will prove that given any $(x_1,x_2,x_3),(y_1,y_2,y_3)\in \partial T$, let
\begin{equation}\label{eq:xyz}
(z_1,z_2,z_3)=(x_1,x_2,x_3)\left(%
\begin{array}{ccc}
  y_2+y_3 & 0 & 0 \\
  y_1 & y_2 & y_3 \\
  y_1 & y_3 & y_2 \\
\end{array}%
\right)
=(y_1,y_2,y_3)\left(%
\begin{array}{ccc}
  x_2+x_3 & 0 & 0 \\
  x_1 & x_2 & x_3 \\
  x_1 & x_3 & x_2 \\
\end{array}%
\right)
,
\end{equation}
then $(z_1,z_2,z_3)\in T$. This result implies the lemma because equations~\eqref{eq:new1}--\eqref{eq:new3} are equivalent to the transformation in~\eqref{eq:xyz}.

We can assume $$(x_1,x_2,x_3)=P_1(\cosh(x),x{\mathrm e}^x,\sinh(x))^{\mathrm T}$$ and $$(y_1,y_2,y_3)=P_2(\cosh(y),y{\mathrm e}^y,\sinh(y))^{\mathrm T},$$
for some $x,y \in [0,w]$, where $P_1$ and $P_2$ are two $3 \times 3$ permutation matrices. This is because the transformation~\eqref{eq:xyz} is linear. Moreover, by symmetry, we w.l.o.g. assume that $x_2\geq x_3$ and $y_2\geq y_3$, (switching $x_2$ and $x_3$ (or $y_2$ and $y_3$) permutes $z_2$ and $z_3$).

Note that when $x=0$ or $y=0$, the lemma is trivially true.

\begin{figure}[h]
\begin{center}
\includegraphics[scale=0.3]{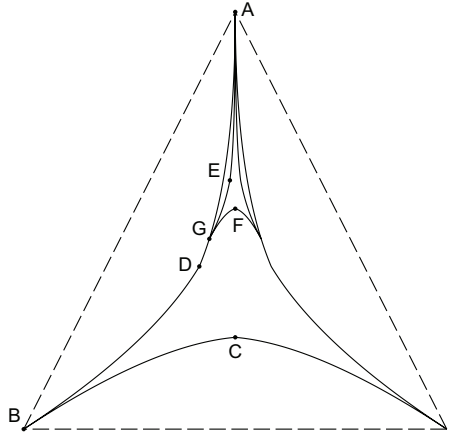}
\caption{Cases $1$, $4$ and $7$. Point $G$ represents $(y_1,y_2,y_3)$. In case $1$, $(x_1,x_2,x_3)$ on curve $AD$ is mapped to $(z_1,z_2,z_3)$ on curve $AE$; in case $4$, $(x_1,x_2,x_3)$ on curve $BD$ is mapped to $(z_1,z_2,z_3)$ on curve $GE$; and in case $7$, $(x_2,x_3,x_1)$ on curve $BC$ is mapped to $(z_1,z_2,z_3)$ on curve $GF$}\label{fig:c147}
\end{center}
\end{figure}

We prove the lemma by considering $9$ cases, depending on $P_1$ and $P_2$. Case $1$, $4$ and $7$ are shown in Figure~\ref{fig:c147}, in which $y_1 \geq y_2 \geq y_3$ and $(y_1,y_2,y_3)$ is represented by point $G$ on curve $AD$. In case $1$, $x_1 \geq x_2 \geq x_3$, $(x_1,x_2,x_3)$ is on curve $AD$, and $(z_1,z_2,z_3)$ is on curve $AE$; in case $4$, $x_2 \geq x_1 \geq x_3$, $(x_1,x_2,x_3)$ is on curve $BD$, and $(z_1,z_2,z_3)$ is on curve $GE$; and in case $7$, $x_1 \geq x_2 \geq x_3$, $(x_2,x_3,x_1)$ is on curve $BC$, and $(z_1,z_2,z_3)$ is on curve $GF$.

\begin{figure}[h]
\begin{center}
\includegraphics[scale=0.3]{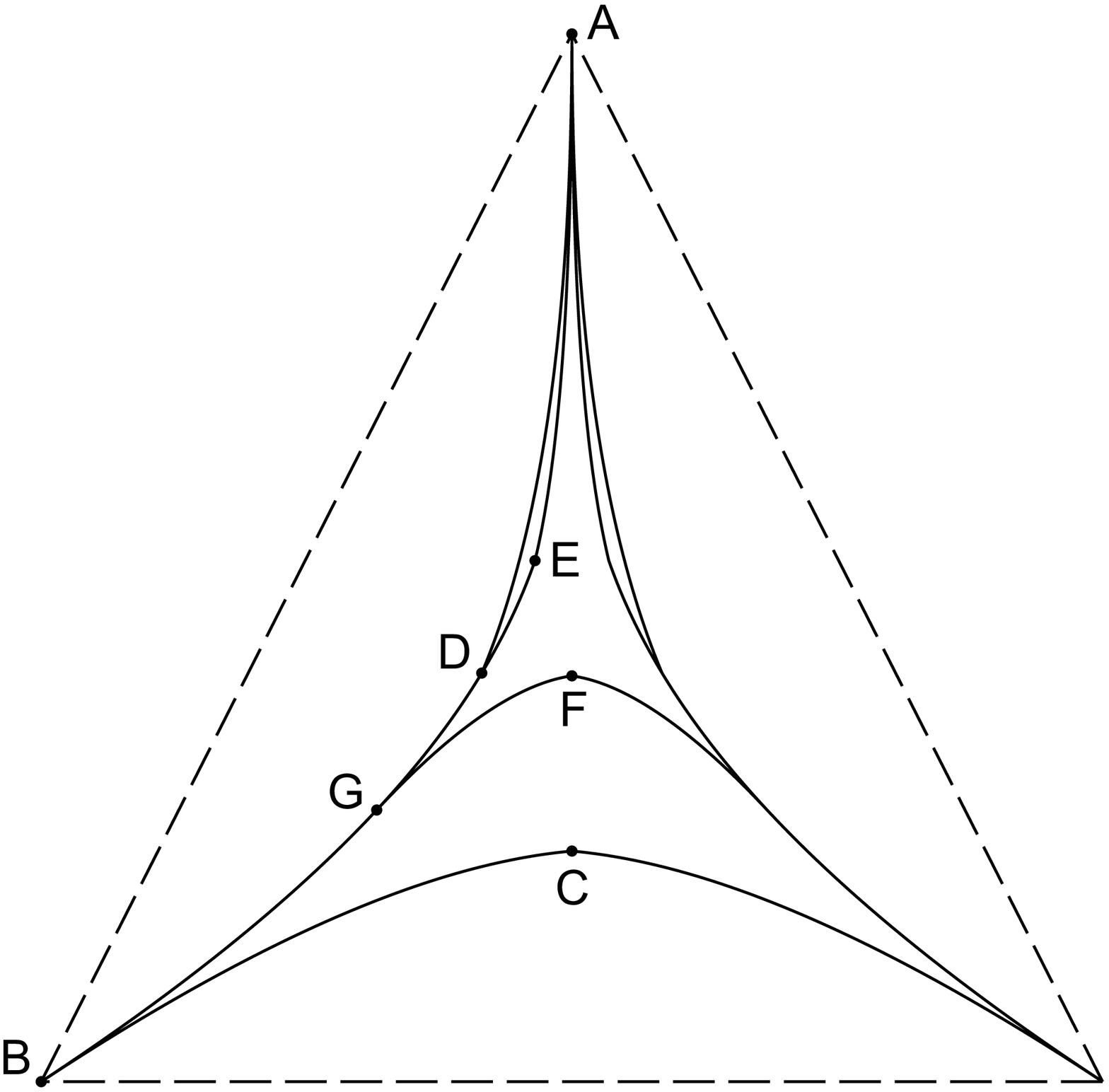}
\caption{Cases $2$, $5$ and $8$. Point $G$ represents $(y_1,y_2,y_3)$. In case $2$, $(x_1,x_2,x_3)$ on curve $AD$ is mapped to $(z_1,z_2,z_3)$ on curve $AE$; in case $5$, $(x_1,x_2,x_3)$ on curve $BD$ is mapped to $(z_1,z_2,z_3)$ on curve $GE$; and in case $8$, $(x_2,x_3,x_1)$ on curve $BC$ is mapped to $(z_1,z_2,z_3)$ on curve $GF$}\label{fig:c258}
\end{center}
\end{figure}

Case $2$, $5$ and $8$ are shown in Figure~\ref{fig:c258}, in which $y_2 \geq y_1 \geq y_3$ and $(y_1,y_2,y_3)$ is represented by point $G$ on curve $BD$. In case $2$, $x_1 \geq x_2 \geq x_3$, $(x_1,x_2,x_3)$ is on curve $AD$, and $(z_1,z_2,z_3)$ is on curve $AE$; in case $5$, $x_2 \geq x_1 \geq x_3$, $(x_1,x_2,x_3)$ is on curve $BD$, and $(z_1,z_2,z_3)$ is on curve $GE$; and in case $8$, $x_1 \geq x_2 \geq x_3$, $(x_2,x_3,x_1)$ is on curve $BC$, and $(z_1,z_2,z_3)$ is on curve $GF$.

\begin{figure}[h]
\begin{center}
\includegraphics[scale=0.3]{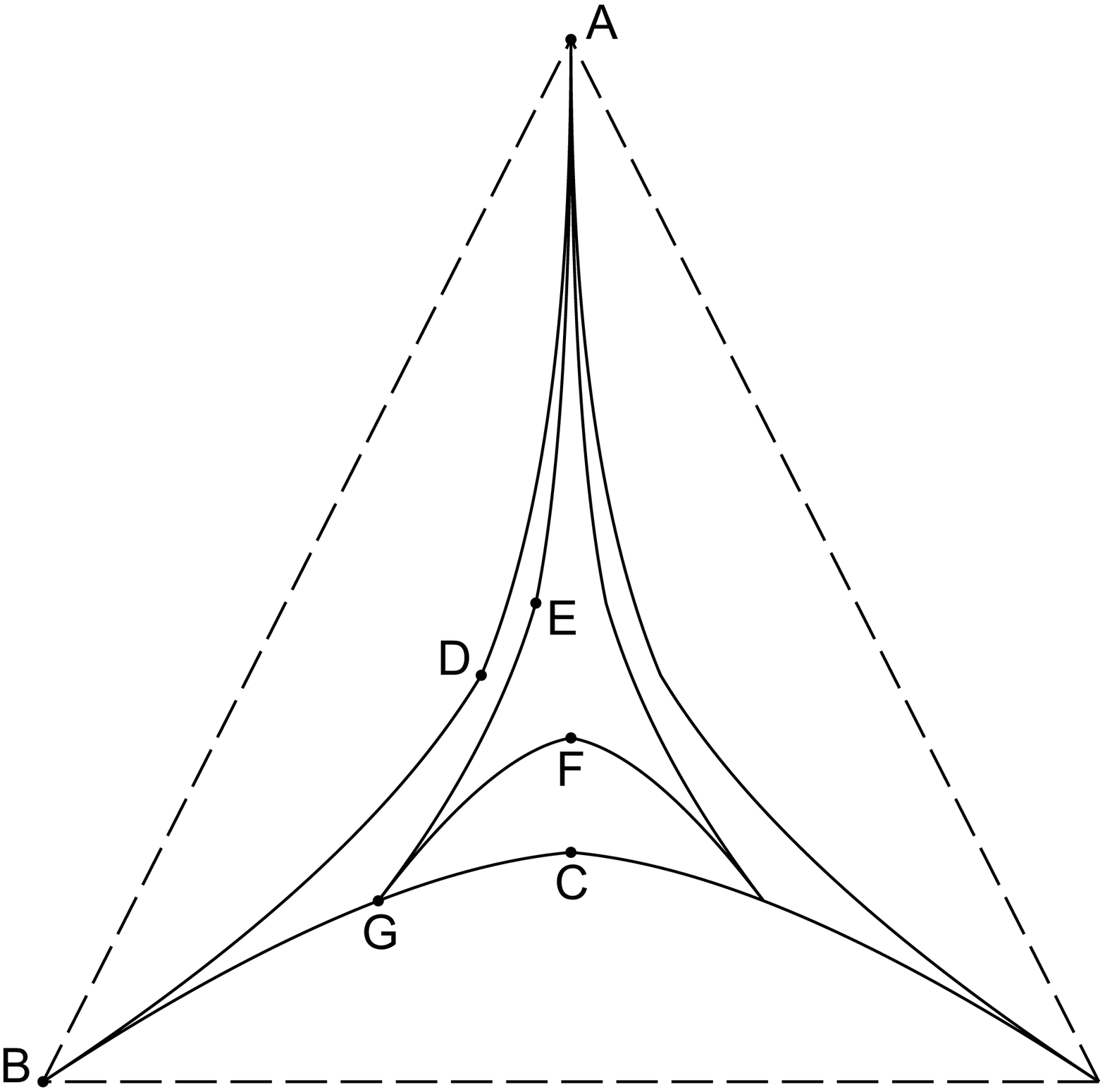}
\caption{Cases $3$, $6$ and $9$. Point $G$ represents $(y_1,y_2,y_3)$. In case $3$, $(x_1,x_2,x_3)$ on curve $AD$ is mapped to $(z_1,z_2,z_3)$ on curve $AE$; in case $6$, $(x_1,x_2,x_3)$ on curve $BD$ is mapped to $(z_1,z_2,z_3)$ on curve $GE$; and in case $9$, $(x_2,x_3,x_1)$ on curve $BC$ is mapped to $(z_1,z_2,z_3)$ on curve $GF$}\label{fig:c369}
\end{center}
\end{figure}

Case $3$, $6$ and $9$ are shown in Figure~\ref{fig:c369}, in which $y_2 \geq y_3 \geq y_1$ and $(y_1,y_2,y_3)$ is represented by point $G$ on curve $BC$. In case $3$, $x_1 \geq x_2 \geq x_3$, $(x_1,x_2,x_3)$ is on curve $AD$, and $(z_1,z_2,z_3)$ is on curve $AE$; in case $6$, $x_2 \geq x_1 \geq x_3$, $(x_1,x_2,x_3)$ is on curve $BD$, and $(z_1,z_2,z_3)$ is on curve $GE$; and in case $9$, $x_1 \geq x_2 \geq x_3$, $(x_2,x_3,x_1)$ is on curve $BC$, and $(z_1,z_2,z_3)$ is on curve $GF$.

Note that in Figure~\ref{fig:c147}, Figure~\ref{fig:c258} and Figure~\ref{fig:c369}, all vectors are projected to the plane $\alpha+\beta+\gamma=1$ and $\alpha,\beta,\gamma \geq 0$.

\begin{description}
\item[Case 1: $x_1\geq x_2\geq x_3$, $y_1\geq y_2\geq y_3$.] We have
\begin{eqnarray*}
 z_1 & = &\sinh(x+y)+y{\mathrm e}^y\cosh(x)+x{\mathrm e}^x\cosh(y),\\
 z_2 & = & xy {\mathrm e}^{x+y} + \sinh(x)\sinh(y),\\
 z_3 & = & x{\mathrm e}^x\sinh(y) + y{\mathrm e}^y\sinh(x).
\end{eqnarray*}
We also have $z_1\geq z_2\geq z_3$, this is because 
$$
z_1-z_2=\sinh(x)\mathrm{e}^{-y}+\cosh(x)\sinh(y)+y\mathrm{e}^y(\cosh(x)-x\mathrm{e}^x)+x\mathrm{e}^x\cosh(y) \geq 0,
$$
and
$$
z_2-z_3=(x\mathrm{e}^x-\sinh(x))(y\mathrm{e}^y-\sinh(y)) \geq 0.
$$

Note that this means $(z_1,z_2,z_3)$ can only intersect $\alpha(\cosh(t),t\mathrm{e}^t,\sinh(t))$, for $\alpha \geq 0$ and $t \in [0,w]$. W.l.o.g., we assume $x\leq y$.

The normal of the surface
$$
\alpha(\cosh(t),t{\mathrm e}^t,\sinh(t)) \quad \alpha \geq 0, t \in [0,w]
$$
at $\alpha=\mu$ and $t=x/2$ is $\vec{n}_1=\frac{\mu}{2} (-{\mathrm e}^{x}+1+x, -2, {\mathrm e}^{x}+x+1)$.
By Lemma~\ref{lem:normal}, we have for every $t \in [0,w]$, $\vec{n}_1 (\cosh(t),t{\mathrm e}^t,\sinh(t)) \leq 0$. To show $(z_1,z_2,z_3) \in T$ it is sufficient to show that $\vec{n}_1 (z_1,z_2,z_3)^{\mathrm T} > 0$ for every $x,y \in (0,w]$.

\begin{multline}\label{ett}
\vec{n}_1(z_1,z_2,z_3)^{\mathrm T} = \\
-\frac{\mu{\mathrm e}^{-x - y}}{4}\Big(
2+x-{\mathrm e}^{2y}
+{\mathrm e}^{3x}({\mathrm e}^{2y}+2x)+2{\mathrm e}^{2x}{\mathrm e}^{2y}xy+2{\mathrm e}^{x}{\mathrm e}^{2y}y\\
-{\mathrm e}^{2x}(2{\mathrm e}^{2y}x^2+3{\mathrm e}^{2y}x+2{\mathrm e}^{2y}y+1)-{\mathrm e}^{x}
\Big).
\end{multline}
We upper bound the expression in the parenthesis on the right-hand
side of~\eqref{ett} by
\begin{equation}\label{ett2}
\begin{split}
\frac{27}{2}x^5+\frac{199}{24}x^4+\frac{15}{2}x^3+\frac{7}{2}x^2\\
-\Big(\frac{4}{3}x^6+\frac{14}{3}x^5+\frac{5}{4}x^4+
\frac{11}{2}x^3+\frac{7}{2}x^2\Big){\mathrm e}^{2y}\\
+\Big(\frac{8}{3}x^5+\frac{3}{2}x^4+\frac{5}{3}x^3+x^2\Big)y \mathrm{e}^{2y},
\end{split}
\end{equation}
which is obtained by using~\eqref{eq:exp4}
on the exponentials involving $x$. We further upper bound~\eqref{ett2} by
\begin{equation}\label{ett4}
\begin{split}
-\frac{x^2}{24}(64x^4y^2-256x^3y^3+64x^4y+96x^3y^2-144x^2y^3+32x^4\\
+160x^3y-12x^2y^2-160xy^3-212x^3+24x^2y
+184xy^2\\-96y^3-169x^2+224yx+120y^2-48x+144y),
\end{split}
\end{equation}
which is obtained by using~\eqref{eq:exp2} on ${\mathrm e}^{2y}$.

Now we show that the expression in the parenthesis of~\eqref{ett4}
is positive. Let $c=x/y$. Note that $c\in (0,1]$. If $c\in [L,U]$ then we
can lower bound the expression in the parenthesis of~\eqref{ett4} by
\begin{equation}\label{ett5}
\begin{split}
64y^6L^4-256y^6U^3+64y^5L^4+96y^5L^3+32y^4L^4-144y^5U^2\\+160y^4L^3
-12y^4U^2-212y^3U^3-160y^4U+24y^3L^2+184y^3L\\
-169y^2U^2-96y^3+224y^2L+120y^2-48yU+144y.
\end{split}
\end{equation}
For $L=0$, $U=5/7$ the lower bound on the expression in the parenthesis of~\eqref{ett4}
becomes
\begin{equation}\label{ett6}
-\frac{32000}{343}y^6-\frac{3600}{49}y^5-\frac{5900}{49}y^4-\frac{59428}{343}y^3+\frac{1655}{49}y^2+\frac{768}{7}y.
\end{equation}
For $L=5/7$, $U=1$ the lower bound on the expression in the parenthesis of~\eqref{ett4}
becomes
\begin{equation}\label{ett7}
-\frac{574656/2401}y^6-\frac{221744/2401}y^5-\frac{252972}{2401}y^4-\frac{8052}{49}y^3+111y^2+96y.
\end{equation}
The polynomials~\eqref{ett6} and~\eqref{ett7} are positive for $y \in(0,64/100)$ (as is
easily checked using Sturm sequences).
\item[Case 2: $x_1\geq x_2\geq x_3$, $y_2\geq y_1\geq y_3$.]
We have
\begin{eqnarray*}
z_1 &=& \cosh(x)\cosh(y)+\cosh(x)\sinh(y)+xy{\mathrm e}^{x+y}+\sinh(x)y{\mathrm e}^y,\\
z_2 &=& x{\mathrm e}^x\cosh(y)+\sinh(x)\sinh(y),\\
z_3 &=& x{\mathrm e}^x\sinh(y)+\sinh(x)\cosh(y).
\end{eqnarray*}
By normalizing $z_1,z_2,z_3$ to $z_1+z_2+z_3=1$, we have
\begin{eqnarray*}
z_1 &=& \frac{-2({\mathrm e}^{x}+{\mathrm e}^{-x}+2xy{\mathrm e}^{x}+y{\mathrm e}^{x}-y{\mathrm e}^{-x})}{-4{\mathrm e}^{x}-2y{\mathrm e}^{x}-4x{\mathrm e}^{x}-4xy{\mathrm e}^{x}+2y{\mathrm e}^{-x}},\\
z_2 &=& \frac{-2x{\mathrm e}^{x}-2x{\mathrm e}^{x-2y}-{\mathrm e}^{x}+{\mathrm e}^{x-2y}+{\mathrm e}^{-x}-{\mathrm e}^{-x-2y}}{-4{\mathrm e}^{x}-2y{\mathrm e}^{x}-4x{\mathrm e}^{x}-4xy{\mathrm e}^{x}+2y{\mathrm e}^{-x}},\\
z_3 &=& \frac{-2x{\mathrm e}^{x}+2x{\mathrm e}^{x-2y}-{\mathrm e}^{x}-{\mathrm e}^{x-2y}+{\mathrm e}^{-x}+{\mathrm e}^{-x-2y}}{-4{\mathrm e}^{x}-2y{\mathrm e}^{x}-4x{\mathrm e}^{x}-4xy{\mathrm e}^{x}+2y{\mathrm e}^{-x}}.
\end{eqnarray*}

We are going to use Lemma~\ref{lem:jacob} to prove that $(z_{1},z_{2},z_{3}) \in T$. We first argue that the image of the boundary (on the boundary at least one of $x=0,x=w,y=0$, and $y=w$ is satisfied) is simple and is in $T$.

The image of $x=0$ is the point $z_{1}=1,z_{2}=0,z_3=0$, which is in $T$. 

The image of $y=0$ is the curve $C_{1}$:
\begin{equation*}
z_{1} = \frac{\cosh(x)}{(1+x){\mathrm e}^{x}}, \quad
z_{2} = \frac{x{\mathrm e}^{x}}{(1+x){\mathrm e}^{x}}, \quad
z_3 = \frac{\sinh(x)}{(1+x){\mathrm e}^{x}}, \quad \mbox{ for } x \in [0,w].
\end{equation*}
Note that $C_{1}  \subseteq \partial S$.

The image of $x=w$ is the curve $C_{2}$:
\begin{eqnarray}
z_{1} &=& \frac{(w+y){\mathrm e}^{w+y}}{(1+y+w){\mathrm e}^{y+w}},\label{eq:dc2_1}\\
z_{2} &=& \frac{\cosh(y+w)}{(1+y+w){\mathrm e}^{y+w}},\nonumber\\
z_3 &=& \frac{\sinh(y+w)}{(1+y+w){\mathrm e}^{y+w}},\nonumber
\end{eqnarray}
for $y \in [0,w]$.
We defer the proof that $C_{2} \subseteq T$ and the only intersection of $C_{1}$ and $C_{2}$ is the end point of both $C_{1}$ and $C_{2}$ to case 5. 

The image of $y=w$ is the curve $C_{3}$:
\begin{eqnarray}
z_{1} &=& \frac{\cosh(x)+w(x{\mathrm e}^{x}+\sinh(x))}{(1+x){\mathrm e}^{x}+w(x{\mathrm e}^{x}+\sinh(x))},\label{eq:dc3_1}\\
z_{2} &=& \frac{\sinh(x)+w(x{\mathrm e}^{x}-\sinh(x))}{(1+x){\mathrm e}^{x}+w(x{\mathrm e}^{x}+\sinh(x))},\nonumber\\
z_3 &=& \frac{x{\mathrm e}^x-w(x{\mathrm e}^{x}-\sinh(x))}{(1+x){\mathrm e}^{x}+w(x{\mathrm e}^{x}+\sinh(x))},\nonumber
\end{eqnarray}
for $x \in [0,w]$. In case 1, we have shown that $C_{3} \subseteq T$ and the only intersection of $C_{3}$ and $C_{1}$ is the end point (which is the point $z_{1}=1,z_{2}=0,z_3=0$) of both $C_{3}$ and $C_{1}$. Note that in~\eqref{eq:dc2_1}, $z_{1}$ is monotonically increasing and $z_{1} \leq 2w/(2w+1)$. Also note that in~\eqref{eq:dc3_1}, $z_1$ is monotonically decreasing and $z_1 \geq 2w/(2w+1)$. Hence, $C_2$ and $C_3$ are simple and have only one intersection which is the end point of both curves. We established that the image of the boundary (which is concatenation of $C_1$, $C_2$, $C_3$) is a simple curve.

We next claim that the Jacobian determinant does not vanish when $x,y \in (0,w)$. Converting from barycentric coordinates we obtain $(Z_1,Z_2)$:
\begin{eqnarray*}
Z_1 &=& \frac{(2x{\mathrm e}^{2x}-{\mathrm e}^{2x}+1){\mathrm e}^{-x-2y}\sqrt{3}}{-4{\mathrm e}^{x}-2y{\mathrm e}^{x}-4x{\mathrm e}^{x}-4xy{\mathrm e}^{x}+2y{\mathrm e}^{-x}},\\
Z_2 &=& \frac{-{\mathrm e}^{x}-3{\mathrm e}^{-x}-4xy{\mathrm e}^{x}-2y{\mathrm e}^{x}+2y{\mathrm e}^{-x}+2x{\mathrm e}^{x}}{-4{\mathrm e}^{x}-2y{\mathrm e}^{x}-4x{\mathrm e}^{x}-4xy{\mathrm e}^{x}+2y{\mathrm e}^{-x}}.
\end{eqnarray*}

\begin{equation*}
\begin{split}
{\rm det}\left(%
\begin{array}{cc}
  \partial Z_1/\partial x & \partial Z_2/\partial x \\
  \partial Z_1/\partial y & \partial Z_2/\partial y
\end{array}%
\right)=-\frac{6\sqrt{3}{\mathrm e}^{-2y-x}(4x{\mathrm e}^{2x}+1+4x^2{\mathrm e}^{2x}-{\mathrm e}^{4x})}
{(-2{\mathrm e}^{x}-y{\mathrm e}^{x}-2x{\mathrm e}^{x}-2xy{\mathrm e}^{x}+y{\mathrm e}^{-x})^3}.
\end{split}
\end{equation*}

For $x\in (0,w)$ we have
$$
4x{\mathrm e}^{2x}+1+4x^2{\mathrm e}^{2x}-{\mathrm e}^{4x}\geq
4x^2-\frac{48}{5}x^3+\frac{40}{3}x^4+\frac{16}{3}x^5 > 0,
$$
using~\eqref{eq:exp3-2} and Sturm sequences.

By Lemma~\ref{lem:jacob}, we have that $(z_1,z_2,z_3) \in T$.

\medskip

\item[Case 3: $x_1\geq x_2\geq x_3$, $y_2\geq y_3\geq y_1$.]
We have
\begin{eqnarray*}
z_1 &=& \cosh(x+y)+y {\mathrm e}^y\cosh(x)+x {\mathrm e}^x \sinh(y),\\
z_2 &=& x {\mathrm e}^x\cosh(y)+y {\mathrm e}^y\sinh(x),\\
z_3 &=& x y {\mathrm e}^{x+y} + \sinh(x)\cosh(y).
\end{eqnarray*}
We also have $z_1 \geq z_2 \geq z_3$, this is because
$$
z_1-z_2=(\cosh(x)-x\mathrm{e}^x)\cosh(y)+y\mathrm{e}^{y-x}+(\sinh(x)+x\mathrm{e}^x)\sinh(y) \geq 0,
$$
and
$$
z_2-z_3=(x\mathrm{e}^x-\sinh(x))(\cosh(y)-y\mathrm{e}^y) \geq 0.
$$
Note that this means $(z_1,z_2,z_3)$ can only intersect $\alpha(\cosh(t),t\mathrm{e}^t,\sinh(t))$, for $\alpha \geq 0$ and $t \in [0,w]$.

The normal of the surface
$$
\alpha(\cosh(t),t{\mathrm e}^t,\sinh(t)) \quad \alpha \geq 0, t \in [0,w]
$$
at $\alpha=\mu$ and $t=x$ is $\vec{n}_2=\frac{\mu}{2} (-{\mathrm e}^{2x}+1+2x, -2, {\mathrm e}^{2x}+2x+1)$. By Lemma~\ref{lem:normal}, we have for every $t \in [0,w]$, $\vec{n}_2 (\cosh(t),t{\mathrm e}^t,\sinh(t)) \leq 0$. To show $(z_1,z_2,z_3) \in T$ it is sufficient to show that $\vec{n}_2 (z_1,z_2,z_3)^{\mathrm T} > 0$ for every $x,y \in (0,w]$.

There are two cases depending on whether $x \geq y$ or $x < y$. If $x \geq y$,
\begin{multline}\label{eq:c21}
\vec{n}_2(z_1,z_2,z_3)^{\mathrm T} = \\
-\frac{\mu{\mathrm e}^{-x - y}}{8}\Big(
-2x-1+4x{\mathrm e}^{2x}+2{\mathrm e}^{2x}-2x{\mathrm e}^{4x}-{\mathrm e}^{4x}+4x^2{\mathrm e}^{2x}\\
+{\mathrm e}^{2y}(1+2x+2y{\mathrm e}^{4x}+2x{\mathrm e}^{4x}+{\mathrm e}^{4x}+4y{\mathrm e}^{2x})\\
-{\mathrm e}^{2y}(2{\mathrm e}^{2x}+4x{\mathrm e}^{2x}+4xy{\mathrm
e}^{4x}+8x^2y{\mathrm e}^{2x}+8xy{\mathrm e}^{2x}+4xy+4x^2{\mathrm
e}^{2x}+6y) \Big).
\end{multline}
We can
upper bound the expression in the parenthesis on the right-hand
side of~\eqref{eq:c21} by
\begin{equation}\label{eq:c22}
\begin{split}
-\frac{2}{3}y\Big(
{\mathrm e}^{4x}(6xy+8xy^3+4xy^2)\\
-{\mathrm e}^{4x}(6+10y^2+9y+8y^3)\\
+{\mathrm e}^{2x}(32x^2y^2+16xy^3+32xy^2+36x^2y+24x^2+36xy+16x^2y^3+24x)\\
-{\mathrm e}^{2x}(16y^3+6y+8y^2)\\
+6+14y^2+12y^3+6xy+15y+4xy^2+8xy^3\Big),
\end{split}
\end{equation}
which is obtained by using~\eqref{eq:exp3-1} on ${\mathrm e}^{2y}$.
To show that~\eqref{eq:c22} is negative it is enough to show
that the expression in the parenthesis of~\eqref{eq:c22} is
positive.
We lower bound the expression in the parenthesis on the
right-hand side of~\eqref{eq:c22} by
\begin{equation}\label{eq:c23}
\begin{split}
-4y^2-12y^3+24x^2+32x^3+16x^4+88x^3y+84x^4y-104x^5y/5\\
+64x^3y^3/3+128x^4y^3/3-416x^5y^3/15+128x^3y^2/3+112x^4y^2/3\\
-1024x^5y^2/15+64x^6y^2+424x^6y/5+416x^6y^3/5-16xy^2-32xy^3\\
+16x^2y^2+48x^2y-16x^2y^3-272x^5/5+112x^6/5+32x^7/5+128x^7y^2/15\\
+48x^7y/5+64x^7y^3/15,
\end{split}
\end{equation}
which is obtained by using~\eqref{eq:exp5} on exponentials involving $x$.

Let $c = y/x$. Note that $c\in (0,1]$. If $c\in [L, U]$ then we
can lower bound~\eqref{eq:c23} by
\begin{equation}\label{eq:c24}
\begin{split}
\frac{1}{15}(
360x^2+480x^3+240x^4+64x^{10}L^3+640x^5L^2+1320x^4L+1272x^7L\\
+640x^7L^3+240x^4L^2+128x^9L^2+720x^3L+560x^6L^2+960x^8L^2\\
+1260x^5L+320x^6L^3+144x^8L+1248x^9L^3-60U^2x^2-180U^3x^3\\
-416x^8U^3-1024x^7U^2-480x^4U^3-240x^3U^2-312x^6U\\
-240x^5U^3-816x^5+336x^6+96x^7 ).
\end{split}
\end{equation}
For $L = 0, U = 4/5$ the lower bound on the expression in the
parenthesis of~\eqref{eq:c24} becomes
\begin{equation}\label{eq:c25}
-\frac{26624}{125}x^8-\frac{13984}{25}x^7+\frac{432}{5}x^6-\frac{23472}{25}x^5-\frac{144}{25}x^4+\frac{5856}{25}x^3+\frac{1608}{5}x^2.
\end{equation}
For $L = 4/5, U = 1$ the lower bound on the expression in the
parenthesis of~\eqref{eq:c24} becomes
\begin{equation}\label{eq:c26}
\begin{split}
\frac{4096}{125}x^{10}+ \frac{90112}{125}x^9+ \frac{1568}{5}x^8+
\frac{10432}{25}x^7+\frac{13656}{25}x^6\\
+\frac{1808}{5}x^5+
\frac{4848}{5}x^4+ 636x^3+ 300x^2.
\end{split}
\end{equation}
The polynomials~\eqref{eq:c25} and~\eqref{eq:c26} are positive
on $(0, 64/100)$ (as is easily checked using Sturm sequences).

If $x \leq y$, then
\begin{multline}\label{eq:c27}
\vec{n}_2(z_1,z_2,z_3)^{\mathrm T} = \\
-\frac{\mu{\mathrm e}^{-x - y}}{8}\Big(
-2x+{\mathrm e}^{2y}-1+2x{\mathrm e}^{2y}-6y{\mathrm e}^{2y}-4xy{\mathrm e}^{2y}\\
+{\mathrm e}^{2x}(4x+2+4x^2+4y{\mathrm e}^{2y})\\
-{\mathrm e}^{2x}(2{\mathrm e}^{2y}+4x{\mathrm e}^{2y}+8x^2y{\mathrm e}^{2y}+8xy{\mathrm e}^{2y}+4x^2{\mathrm e}^{2y})\\
+{\mathrm e}^{4x}(2y{\mathrm e}^{2y}+2x{\mathrm e}^{2y}+{\mathrm e}^{2y})\\
-{\mathrm e}^{4x}(1+4xy{\mathrm e}^{2y}+2x) \Big).
\end{multline}

We can
upper bound the expression in the parenthesis on the right-hand
side of~\eqref{eq:c27} by
\begin{equation}\label{eq:c28}
\begin{split}
-\frac{4}{15}x^2\Big(
{\mathrm e}^{2y}(60y+140xy+170x^2y+8x^5y+84x^3y+156x^4y+4x^5\\
-30x-114x^4-65x^2-112x^3)\\
+30x+65x^2+78x^3+46x^4-8x^5 \Big),
\end{split}
\end{equation}
which is obtained by using~\eqref{eq:exp5} on exponentials involving $x$.
To show that~\eqref{eq:c28} is negative it is enough to show
that the expression in the parenthesis of~\eqref{eq:c28} is
positive.

The polynomial
\begin{equation*}
\begin{split}
60y+140xy+170x^2y+8x^5y+84x^3y+156x^4y+4x^5\\
-30x-114x^4-65x^2-112x^3
\end{split}
\end{equation*}
is lower bounded by
\begin{equation}\label{eq:c29}
\begin{split}
60x+140x^2+170x^3+8x^6+84x^4+156x^5+4x^5\\
-30x-114x^4-65x^2-112x^3,
\end{split}
\end{equation}
since $x \leq y$. \eqref{eq:c29} is positive for $x \in
(0,64/100]$ using Sturm sequences. The polynomial $$30x+65x^2+78x^3+46x^4-8x^5$$ is
positive for $x \in (0,64/100]$ using Sturm sequences. Hence the expression in the
parenthesis of ~\eqref{eq:c28} is positive.
\item[Case 4: $x_2\geq x_1\geq x_3$, $y_1\geq y_2\geq y_3$.] This case is equivalent to case 2.
\item[Case 5: $x_2\geq x_1\geq x_3$, $y_2\geq y_1\geq y_3$.]
We have
\begin{eqnarray*}
 z_1 & = & z{\mathrm e}^{x+y},\\
 z_2 & = & \cosh(x+y),\\
 z_3 & = & \sinh(x+y).
\end{eqnarray*}
If $x+y \leq w$, then we are done, since $z_2 \geq z_1 \geq z_3$ and $(z_1,z_2,z_3) \in \partial T$.
We next assume that $x+y = w+b$ for some $b \in (0,w]$. Hence, $z_1 \geq z_2 \geq z_3$. Note that this means that $(z_1,z_2,z_3)$ can only intersect $\alpha(\cosh(t),t\mathrm{e}^t,\sinh(t))$, for $\alpha \geq 0$ and $t \in [0,w]$.

The normal of the surface
$$
\alpha(\cosh(t),t{\mathrm e}^t,\sinh(t)) \quad \alpha \geq 0, t \in [0,w]
$$
at $\alpha=\mu$ and $t=a$ is $\vec{n}_3=\frac{\mu}{2} (-{\mathrm e}^{2a}+1+2a, -2, {\mathrm e}^{2a}+2a+1)$. By Lemma~\ref{lem:normal}, we have for every $t \in [0,w]$, $\vec{n}_3 (\cosh(t),t{\mathrm e}^t,\sinh(t)) \leq 0$. To show that $(z_1,z_2,z_3)$ does not intersect $\alpha(\cosh(t),t{\mathrm e}^t,\sinh(t))$ we will consider two cases (depending on the value of $b$) and use a different $\vec{n}_3$ in each case. For $b \in (0,1/10]$ we use $\vec{n}_3$ with $a=w$ and for $b \in [1/10,w]$ we use $a=1/2$.

We have
\begin{multline}\label{eq:c51}
\vec{n}_3(z_1,z_2,z_3)^{\mathrm T} = \\
\frac{\mu{\mathrm e}^{-w - b}}{4(2w-1)}\Big({\mathrm e}^{2b}(2w+{\mathrm e}^{2a}+2b+2a+4ab+4aw)\\
-{\mathrm e}^{2b}(2b{\mathrm e}^{2a}+1+2w{\mathrm e}^{2a})\\
-4aw+2a+{\mathrm e}^{2a}-6w+3-2w{\mathrm e}^{2a}
\Big),
\end{multline}
where we use the equation ${\mathrm e}^{2w}=1/(2w-1)$.

We lower bound the expression in the parenthesis on the right-hand side
of~\eqref{eq:c51} by
\begin{equation}\label{eq:c52}
\begin{split}
\frac{1}{3}\Big(
b^5(8a-8{\mathrm e}^{2a}+4)\\
+b^4(4+4w-6{\mathrm e}^{2a}+8aw-8w{\mathrm e}^{2a}+20a)\\
+b^3(32a+8+8w-8w{\mathrm e}^{2a}+16aw-8{\mathrm e}^{2a})\\
+b^2(6+36a+24aw+12w-6{\mathrm e}^{2a}-12w{\mathrm e}^{2a})\\
+b(-12w{\mathrm e}^{2a}+24a+24aw+12w)\\
+6+6{\mathrm e}^{2a}+12a-12w{\mathrm e}^{2a}-12w
\Big),
\end{split}
\end{equation}
which is obtained by using~\eqref{eq:exp4}
on ${\mathrm e}^{2b}$.

If $a=w$, then by applying ${\mathrm e}^{2w}=1/(2w-1)$, \eqref{eq:c52} is equal to
\begin{equation}\label{eq:c53}
\begin{split}
\frac{2b}{3(2w-1)}\Big(
b^4(8w^2-6)\\
+b^3(20w^2-5+8w^3-12w)\\
+b^2(32w^2+16w^3-8-16w)\\
+b(-6+24w^3+36w^2-24w)\\
-24w+24w^2+24w^3
\Big).
\end{split}
\end{equation}

By lower bounding the polynomials involving $w$ in~\eqref{eq:c53}, we further lower bound the expression in the parenthesis of~\eqref{eq:c53} by
\begin{equation*}
\frac{1}{100}\Big(
-274b^4-241b^3-98b^2-37b+73
\Big),
\end{equation*}
which is positive when $b \in (0,1/10]$, using Sturm sequences.

If $a=1/2$, then \eqref{eq:c52} is equal to
\begin{equation}\label{eq:c55}
\begin{split}
\frac{2}{3}\Big(
b^5(-4{\mathrm e}+4)\\
+b^4(-4{\mathrm e}w+4w+7-3{\mathrm e})\\
+b^3(12+8w-4{\mathrm e}w-4{\mathrm e})\\
+b^2(12w-6{\mathrm e}w+12-3{\mathrm e})\\
+b(6+12w-6{\mathrm e}w)\\
6+3{\mathrm e}-6w-6{\mathrm e}w
\Big).
\end{split}
\end{equation}

By lower bounding the polynomial involving $w$ and ${\mathrm e}$, we further lower bound the expression in the parenthesis of~\eqref{eq:c55} by
\begin{equation*}
\frac{1}{100}\Big(
-688b^5-555b^4-71b^3+109b^2+324b-11
\Big),
\end{equation*}
which is positive when $b \in [1/10,64/100]$, using Sturm sequences.
\item[Case 6: $x_2\geq x_1\geq x_3$, $y_2\geq y_3\geq y_1$.]
We have
\begin{eqnarray*}
z_{1} &=& x{\mathrm e}^{x}\cosh(y)+xy{\mathrm e}^{x+y}+\cosh(x)\sinh(y)+\sinh(x)\sinh(y),\\
z_{2} &=& \cosh(x)\cosh(y)+\sinh(x)y{\mathrm e}^{y},\\
z_{3} &=& \cosh(x)y{\mathrm e}^{y}+\sinh(x)\cosh(y).
\end{eqnarray*}
By normalizing $z_1,z_2,z_3$ to $z_1+z_2+z_3=1$, we have
\begin{eqnarray*}
z_{1} &=& \frac{2x{\mathrm e}^{y}+2x{\mathrm e}^{-y}+4xy{\mathrm e}^{y}+2{\mathrm e}^{y}-2{\mathrm e}^{-y}}{4{\mathrm e}^{y}+4y{\mathrm e}^{y}+2x{\mathrm e}^{y}+4xy{\mathrm e}^{y}+2x{\mathrm e}^{-y}},\\
z_{2} &=& \frac{{\mathrm e}^{y}+{\mathrm e}^{-y}+{\mathrm e}^{y-2x}+{\mathrm e}^{-2x-y}+2y{\mathrm e}^{y}-2y{\mathrm e}^{y-2x}}{4{\mathrm e}^{y}+4y{\mathrm e}^{y}+2x{\mathrm e}^{y}+4xy{\mathrm e}^{y}+2x{\mathrm e}^{-y}},\\
z_{3} &=& \frac{2y{\mathrm e}^{y}+2y{\mathrm e}^{y-2x}+{\mathrm e}^{y}+{\mathrm e}^{-y}-{\mathrm e}^{y-2x}-{\mathrm e}^{-2x-y}}{4{\mathrm e}^{y}+4y{\mathrm e}^{y}+2x{\mathrm e}^{y}+4xy{\mathrm e}^{y}+2x{\mathrm e}^{-y}}.
\end{eqnarray*}

We are going to use Lemma~\ref{lem:jacob} to prove that $(z_{1},z_{2},z_{3}) \in T$. We first argue that the image of the boundary (on the boundary at least one of $x=0,x=w,y=0$, and $y=w$ is satisfied) is simple and is in $T$.

The image of $x=0$ is the curve $C_4$:
\begin{equation*}
z_1 = \frac{\sinh(y)}{(y+1){\mathrm e}^y}, \quad
z_2 = \frac{\cosh(y)}{(y+1){\mathrm e}^y}, \quad
z_3 = \frac{y{\mathrm e}^y}{(y+1){\mathrm e}^y}, \quad \mbox{ for } y \in [0,w].
\end{equation*}
Note that $C_4 \subseteq \partial T$. 

The image of $y=0$ is the curve $C_5$:
\begin{equation*}
z_1 = \frac{x{\mathrm e}^x}{(x+1){\mathrm e}^x}, \quad
z_2 = \frac{\cosh(x)}{(x+1){\mathrm e}^x}, \quad
z_3 = \frac{\sinh(x)}{(x+1){\mathrm e}^x}, \quad \mbox{ for } x \in [0,w].
\end{equation*}
Note that $C_5 \subseteq \partial T$. The only intersection of $C_4$ and $C_5$ is point $(0,1,0)$ which is the end point of both $C_4$ and $C_5$. 

The image of $x=w$ is the curve $C_6$:
\begin{eqnarray}
z_1 &=& \frac{\sinh(y)+w(\cosh(y)+y{\mathrm e}^y)}{(1+y){\mathrm e}^y + w(\cosh(y)+y{\mathrm e}^y)},\label{eq:dc6_1}\\
z_2 &=& \frac{y{\mathrm e}^y+w(\cosh(y)-y{\mathrm e}^y)}{(1+y){\mathrm e}^y + w(\cosh(y)+y{\mathrm e}^y)},\label{eq:dc6_2}\\
z_3 &=& \frac{\cosh(y)-w(\cosh(y)-y{\mathrm e}^y)}{(1+y){\mathrm e}^y + w(\cosh(y)+y{\mathrm e}^y)},\label{eq:dc6_3}
\end{eqnarray}
for $y \in [0,w]$. The fact that $C_6 \subseteq T$ and that $C_6$ and $C_4$ are disjoint follows known from case 3. Note that in $C_5$, $z_1 \leq w/(1+w)$. Also note that in~\eqref{eq:dc6_1}, $z_1$ is monotonically increasing and $z_1 \geq w/(1+w)$. Hence, $C_6$ is simple and the only intersection of $C_5$ and $C_6$ is the end point of both $C_5$ and $C_6$. 

The image of $y=w$ is the curve $C_7$:
\begin{equation*}
z_1 = \frac{2wx+1-w}{2wx+1+w}, \quad
z_2 = \frac{w}{2wx+1+w}, \quad
z_3 = \frac{w}{2wx+1+w},
\end{equation*}
for $x \in [0,w]$. We have $C_7 \subseteq T$ since $z_2=z_3$ and $z_1 \geq (1-w)/(1+w)$. The curves $C_7$ and $C_5$ are disjoint. From~\eqref{eq:dc6_2} and~\eqref{eq:dc6_3}, we know that in $C_6$, $z_2 > z_3$ for $y \in [0,w)$. Thus the only intersection of $C_6$ and $C_7$ is the end point of both $C_6$ and $C_7$. We established that the image of the boundary (which is concatenation of $C_4$, $C_5$, $C_6$, and $C_7$) is a simple curve.

We next claim that the Jacobian determinant does not vanish when $x,y \in (0,w)$. Converting from barycentric coordinates we obtain $(Z_{1},Z_{2})$:
\begin{eqnarray*}
Z_{1} &=& \frac{- \sqrt{3}{\mathrm e}^{-2x+y} - \sqrt{3}{\mathrm e}^{-2x-y} +2\sqrt{3} y {\mathrm e}^{-2 x+y} }{2(2 {\mathrm e}^y+2 y
{\mathrm e}^y+x {\mathrm e}^y +2 x y{\mathrm e}^y +x {\mathrm
e}^{-y})},\\
Z_{2} &=& \frac{2 x {\mathrm e}^y +2 x {\mathrm e}^{-y}+4 x y {\mathrm
e}^y +{\mathrm e}^y-3 {\mathrm e}^{-y}-2 y {\mathrm e}^y}{2(2 {\mathrm e}^y+2 y
{\mathrm e}^y+x {\mathrm e}^y +2 x y{\mathrm e}^y +x {\mathrm
e}^{-y})}.
\end{eqnarray*}

\begin{equation*}
\begin{split}
{\rm det}\left(%
\begin{array}{cc}
  \partial Z_1/\partial x & \partial Z_2/\partial x \\
  \partial Z_1/\partial y & \partial Z_2/\partial y
\end{array}%
\right)=
\frac{-6\sqrt{3}{\mathrm e}^{-2x-y} (4y^2 {\mathrm e}^{2y}+4y{\mathrm e}^{2y}+{\mathrm e}^{4y}-1)}
{(2{\mathrm e}^{y}+2y{\mathrm e}^{y}+x{\mathrm e}^{y}+2xy{\mathrm e}^{y}+x{\mathrm e}^{-y})^3}.
\end{split}
\end{equation*}

For $y\in(0,w)$, we have
$$
4y^2 {\mathrm e}^{2y}+4y{\mathrm e}^{2y}+{\mathrm e}^{4y}-1 > 0,
$$
since $\mathrm{e}^{4y} > 1$ for $y > 0$.

By Lemma~\ref{lem:jacob}, we have that $(z_1,z_2,z_3) \in T$.
\item[Case 7: $x_2\geq x_3\geq x_1$, $y_1\geq y_2\geq y_3$.] This case is equivalent to case 3.
\item[Case 8: $x_2\geq x_3\geq x_1$, $y_2\geq y_1\geq y_3$.] This case is equivalent to case 6.
\item[Case 9: $x_2\geq x_3\geq x_1$, $y_2\geq y_3\geq y_1$.]
We have
\begin{eqnarray*}
 z_1 & = &\sinh(x+y)+y{\mathrm e}^y\sinh(x)+x{\mathrm e}^x\sinh(y),\\
 z_2 & = & xy {\mathrm e}^{x+y} + \cosh(x)\cosh(y),\\
 z_3 & = & x{\mathrm e}^x\cosh(y) + y{\mathrm e}^y\cosh(x).
\end{eqnarray*}
We also have $z_2\geq z_3$, this is because
$$
z_2-z_3=(\cosh(x)-x\mathrm{e}^x)(\cosh(y)-y\mathrm{e}^y) \geq 0.
$$

We first show that $(z_1,z_2,z_3)$ does not extend beyond the boundary defined by $\alpha(\sinh(t),\cosh(t),t{\mathrm e}^t)$ for $\alpha \geq 0, t \in [0,w]$. W.l.o.g., we assume that $x\leq y$. The normal of the surface
$$
\alpha(\sinh(t),\cosh(t),t{\mathrm e}^t) \quad \alpha \geq 0, t \in [0,w]
$$
at $\alpha=\mu$ and $t=y$ is $\vec{n}_4=\frac{\mu}{2} (2y+{\mathrm e}^{2y}+1, -{\mathrm e}^{2y}+1+2y, -2)$. By Lemma~\ref{lem:normal}, we have for every $t \in [0,w]$, $\vec{n}_4 (\sinh(t),\cosh(t),t{\mathrm e}^t) \leq 0$. To show $(z_1,z_2,z_3) \in T$ it is sufficient to show that $\vec{n}_4 (z_1,z_2,z_3)^{\mathrm T} > 0$ for every $x,y \in (0,w]$.

\begin{multline}\label{dtt}
\vec{n}_4(z_1,z_2,z_3)^{\mathrm T} = \\
-\frac{\mu{\mathrm e}^{-x-y}}{8}\Big(
{\mathrm e}^{2x}(
4{\mathrm e}^{4y}xy
+4{\mathrm e}^{2y}x
+4xy
+6x
)\\
-{\mathrm e}^{2x}(
{\mathrm e}^{4y}
+2{\mathrm e}^{4y}y
+2{\mathrm e}^{4y}x
+4{\mathrm e}^{2y}y
+4{\mathrm e}^{2y}y^2
+8{\mathrm e}^{2y}xy^2
+8{\mathrm e}^{2y}xy
+2{\mathrm e}^{2y}
+2y
+1
)\\
+2{\mathrm e}^{4y}y
+{\mathrm e}^{4y}
+4{\mathrm e}^{2y}y^2
+4{\mathrm e}^{2y}y
+2{\mathrm e}^{2y}
+2y
+1
\Big).
\end{multline}
We can upper bound
the expression in the parenthesis on the right-hand side of~\eqref{dtt}
by
\begin{equation}\label{dtt2}
\begin{split}
-2x\Big(
{\mathrm e}^{4y}(-8x^2y+2x^2-2xy+3x+2)\\
+{\mathrm e}^{2y}(8x^2y^2+8x^2y+12xy^2-8x^2+12xy+8y^2-2x+8y)\\
-8x^2y-12x^2-2xy-5x-2\Big),
\end{split}
\end{equation}
which is obtained using~\eqref{eq:exp2}
on ${\mathrm e}^{2x}$.

To show that~\eqref{dtt2} is negative it is enough to show that the expression in
the parenthesis of~\eqref{dtt2} is positive.

The polynomial $$-8x^2y+2x^2-2xy+3x+2$$
is positive for $x,y\in (0,w]$, this can be seen by plugging-in $y=64/100$ (since  $y$'s
occur only with negative coefficients) and using Sturm sequences. The polynomial
$$8x^2y^2+8x^2y+12xy^2-8x^2+12xy+8y^2-2x+8y$$ is positive for $x,y\in (0,w]$, $x\leq y$
since $8y\geq 2x$ and $8y^2\geq 8x^2$. Hence we can lower bound the expression
in the parenthesis of~\eqref{dtt2} by
\begin{equation}\label{dtt3}
\begin{split}
16x^2y^3-8x^2y^2+24xy^3-16x^2y+28xy^2+16y^3\\
-18x^2+16xy+24y^2-4x+16y,
\end{split}
\end{equation}
where we used~\eqref{eq:exp1} to lower bound $\mathrm{e}^{2y}$ and $\mathrm{e}^{4y}$.

We can bound~\eqref{dtt3} from below by
\begin{equation*}
\begin{split}
0 y^3-8y^4+0y^3-16y^3+0y^2+16y^3-18y^2+0y+24y^2-4y+16y\\
 = -8y^4+6y^2+12y,
\end{split}
\end{equation*}
which is positive for $y\in (0,w]$, using Sturm sequences.

Next we show that $(z_1,z_2,z_3) \in T$. If $y=w$, then $y_2=y_3=w\mathrm{e}^w$ and $y_1=\sinh(w)$. Hence, $z_2=z_3$ and
\begin{equation*}
\frac{z_2}{z_1+z_2+z_3}=\frac{w\mathrm{e}^w(x\mathrm{e}^x+\cosh(x))}{(1+x)\mathrm{e}^x(1+w)\mathrm{e}^w-\sinh(x)\sinh(w)} \leq \frac{w}{1+w} = \frac{y_2}{y_1+y_2+y_3}.
\end{equation*}
Then we have $(z_1,z_2,z_3) \in T$.

Fix $y \in (0,w)$. Define a mapping 
\begin{equation}\label{eq:m00}
\vec{s} \mapsto \frac{1}{y_1+y_2+y_3-\vec{s}_1 y_1}\vec{s}\left(%
\begin{array}{ccc}
  y_2+y_3 & 0 & 0 \\
  y_1 & y_2 & y_3 \\
  y_1 & y_3 & y_2 \\
\end{array}%
\right),
\end{equation}
where $\vec{s} \in S$. Note that the matrix in~\eqref{eq:m00} is non-sigular (and hence the mapping is injective; to obtain the preimage of a vector we multiply by the inverse of the matrix on the right and normalize the entries to sum to one). Thus, we have that the image (under the mapping~\eqref{eq:m00}) of a simple curve is a simple curve.

Let $R$ contains the vectors $\vec{r} \in S$ such that $\vec{r}_2 \geq \vec{r}_3$. The boundary of $R$ is a simple curve which is the concatenation of the following four curves:
\begin{eqnarray*}
C_8:&&\left(\frac{\cosh(a)}{(1+a)\mathrm{e}^a},\frac{a\mathrm{e}^a}{(1+a)\mathrm{e}^a},\frac{\sinh(a)}{(1+a)\mathrm{e}^a}\right), \quad a \in [0,w],\\
C_9:&&\left(\frac{a\mathrm{e}^a}{(1+a)\mathrm{e}^a},\frac{\cosh(a)}{(1+a)\mathrm{e}^a},\frac{\sinh(a)}{(1+a)\mathrm{e}^a}\right), \quad a \in [0,w],\\
C_{10}:&&\left(\frac{\sinh(a)}{(1+a)\mathrm{e}^a},\frac{\cosh(a)}{(1+a)\mathrm{e}^a},\frac{a\mathrm{e}^a}{(1+a)\mathrm{e}^a}\right), \quad a \in [0,w],\\
C_{11}:&&(1-2a,a,a), \quad a \in [0,w/(1+w)].
\end{eqnarray*}
By case 3 and case 6, we have the image of $C_8$ and $C_9$ is in $R$ and is a simple curve connecting $(1,0,0)$ and 
$$\frac{1}{y_1+y_2+y_3}(y_1,y_2,y_3) \in C_{10}.$$ Hence, the image of $C_8$ and $C_9$ divides $R$ into two regions $R_1$ and $R_2$. We assume w.l.o.g. that $C_{11} \subseteq R_2$. The image of $C_{11}$ is the segment
$$
(1-2a,a,a), \quad a \in \left[0,\frac{(y_2+y_3)w}{2wy_1+(y_2+y_3)(w+1)}\right],
$$
which is a portion of $C_{11}$ and thus is in $R_2$.
Hence, the image of $C_{10}$ is a curve connecting
$$\frac{1}{y_1+y_2+y_3}(y_1,y_2,y_3)$$ and a vector on $C_{11}$. Since the image of $C_{10}$ does not extend beyond $C_{10}$, the image of $C_{10}$ is in $R_2$. Hence, we showed that the image of $C_{10}$ is in $S$.
\end{description}

\end{proof}

\begin{proof}[Proof of Theorem~\ref{thm:gg_neg}]
For the sake of contradiction, we suppose there are two vectors $\vec{s},\vec{t} \in S$ such that after the $2$-glue, the new vector $\vec{r}$ is not in $S$.

By using $\vec{s}_3=1-\vec{s}_1-\vec{s}_2$ and $\vec{t}_3=1-\vec{t}_1-\vec{t}_2$, \eqref{eq:new1}--\eqref{eq:new2} can be simplified as
\begin{eqnarray}
\vec{r}_1 &=& \frac{\vec{s}_1+\vec{t}_1-2\vec{s}_1\vec{t}_1}{1-\vec{s}_1\vec{t}_1},\label{eq:tn1}\\
\vec{r}_2 &=& \frac{\vec{s}_2+\vec{t}_2-\vec{t}_1\vec{s}_2-\vec{s}_1\vec{t}_2-2\vec{s}_2\vec{t}_2}{1-\vec{s}_1\vec{t}_1}.\label{eq:tn2}
\end{eqnarray}
Note that $\vec{r}_1,\vec{r}_2$ can be viewed as functions of $\vec{s}_1,\vec{s}_2$, and the Jacobian determinant is
\begin{equation}\label{eq:det11}
{\rm det}\left(%
\begin{array}{cc}
  \partial \vec{r}_1/\partial \vec{s}_1 & \partial \vec{r}_2/\partial \vec{s}_1 \\
  \partial \vec{r}_1/\partial \vec{s}_2 & \partial \vec{r}_2/\partial \vec{s}_2
\end{array}%
\right)=\frac{(\vec{t}_1-1)^2(\vec{t}_1+2\vec{t}_2-1)}{(1-\vec{s}_1\vec{t}_1)^3}.
\end{equation}
Note that when $\vec{t}_1 \neq 1$ and $\vec{t}_1+2\vec{t}_2 \neq 1$, \eqref{eq:det11} is non-zero for all $\vec{s}_1,\vec{s}_2 \in (0,1)$. (If $\vec{r}_1$, $\vec{r}_2$ are viewed as functions of $\vec{t}_1$, $\vec{t}_2$, one obtains the same expression for the Jacobian determinant (with the roles of $\vec{s}_1$, $\vec{s}_2$ and $\vec{t}_1$, $\vec{t}_2$ switched).)

W.l.o.g., we assume that $\vec{r}_1 < \vec{r}_2$ and $\vec{r}_1 < \vec{r}_3$. Starting from $\vec{r}$, we move $\vec{r}$ in the direction $\Delta \vec{r}=(-1,0,0)$. Note that this keeps $\vec{r} \not \in S$. We can always move $\vec{r}$ (in the direction $\Delta \vec{r}$) by moving $\vec{s}$ until $\vec{s}$ hits the boundary $\partial S$ or the Jacobian determinant (of $\vec{r}_1$, $\vec{r}_2$ viewed as functions of $\vec{s}_1$, $\vec{s}_2$) is zero. Then we move $\vec{r}$ (again in the direction $\Delta \vec{r}$) by moving $\vec{t}$ until $\vec{t}$ hits the boundary $\partial S$ or the Jacobian determinant (of $\vec{r}_1$, $\vec{r}_2$ viewed as functions of $\vec{t}_1$, $\vec{t}_2$) is zero. In this way, we find vectors $\vec{s}'$ and $\vec{t}'$ such that $\vec{r}' \not\in S$ and either both $\vec{s}'$ and $\vec{t}'$ are on $\partial S$ or the Jacobian determinant of $\vec{s}'$ or $\vec{t}'$ is zero. We next show that this can not happen.

The case that both $\vec{s}'$ and $\vec{t}'$ are on $\partial S$ cannot happen because of Lemma~\ref{lem:boundaries}.

Now we assume that the Jacobian determinant of $\vec{s}'$ or $\vec{t}'$ vanishes. W.l.o.g., we assume that $\vec{s}'_1=1$ or $\vec{s}'_1+2\vec{s}'_2=1$.
If $\vec{s}'_1 = 1$ then $\vec{r}'=(1,0,0) \in S$ which contradicts the assumption that $\vec{r}' \not\in S$. If $\vec{s}'_1+2\vec{s}'_2=1$, then we have $\vec{s}'_2=\vec{s}'_3$. By~\eqref{eq:new1}--\eqref{eq:new3}, we have $\vec{r}'_2=\vec{r}'_3$ and $\vec{r}'_1 \geq \vec{s}'_1$, which implies that $\vec{r}' \in S$, a contradiction.

This completes the proof.

\end{proof}

\bibliographystyle{plain}
\bibliography{bibfile}

\end{document}